\documentclass[lettersize,journal]{IEEEtran}
\IEEEoverridecommandlockouts
\ifCLASSINFOpdf
  \usepackage[pdftex]{graphicx}
  \usepackage{epstopdf}
\else
  \usepackage[dvips]{graphicx}
\fi

\usepackage[hyphens]{url}
\usepackage{setspace}
\usepackage[cmex10]{amsmath}
\usepackage{bbm}
\usepackage{amssymb}
\usepackage[usenames,dvipsnames]{color}
\usepackage{graphicx}
\usepackage{float}
\usepackage[caption=false,font=footnotesize]{subfig}
\usepackage{wrapfig}
\usepackage{algorithmic}
\usepackage[algoruled, linesnumbered]{algorithm2e} 
\usepackage{mathrsfs}
\usepackage{cite}
\usepackage{array}
\usepackage{mdwmath}
\usepackage{mdwtab}
\usepackage{fixltx2e}
\usepackage{amsthm}
\usepackage[margin=0.7in]{geometry}
\usepackage{tikz}
\usepackage{tablefootnote}
\usepackage{multirow}
\allowdisplaybreaks

\newtheorem{lemma}{Lemma}

\newtheorem{remark}{Remark}
\DeclareMathOperator*{\argmax}{argmax}

\begin{document}
%
\title{WiDRa - Enabling Millimeter-Level Differential Ranging Accuracy in Wi-Fi Using Carrier Phase}
\author{Vishnu V. Ratnam, \IEEEmembership{Senior Member,~IEEE}, Bilal Sadiq, \IEEEmembership{Member,~IEEE}, Hao Chen, \IEEEmembership{Member,~IEEE}, \\ Wei Sun, \IEEEmembership{Member,~IEEE}, Shunyao Wu, \IEEEmembership{Member,~IEEE}, Boon L. Ng, \IEEEmembership{Member,~IEEE}, \\ Jianzhong (Charlie) Zhang, \IEEEmembership{Fellow,~IEEE}
\thanks{All authors are with the Standards and Mobility Innovation Lab, Samsung Research America,
Plano, Texas, USA. (e-mail: ratnamvishnuvardhan@gmail.com).}
}

\maketitle

\begin{abstract}
Although Wi-Fi is an ideal technology for many ranging applications, the performance of current methods is limited by the system bandwidth, leading to low accuracy of $\sim 1$ m. For many applications, measuring differential range, viz., the change in the range between adjacent measurements, is sufficient. Correspondingly, this work proposes WiDRa - a Wi-Fi based Differential Ranging solution that provides differential range estimates by using the sum-carrier-phase information. The proposed method is not limited by system bandwidth and can track range changes even smaller than the carrier wavelength. The proposed method is first theoretically justified, while taking into consideration the various hardware impairments affecting Wi-Fi chips. In the process, methods to isolate the sum-carrier phase from the hardware impairments are proposed. Extensive simulation results show that WiDRa can achieve a differential range estimation root-mean-square-error (RMSE) of $\approx 1$ mm in channels with a Rician-factor $\geq 7$ (a $100 \times$ improvement to existing methods). The proposed methods are also validated on off-the-shelf Wi-Fi hardware to demonstrate feasibility, where they achieve an RMSE of $< 1$ mm in the differential range. Finally, limitations of current investigation and future directions of exploration are suggested, to further tap into the potential of WiDRa. 
\end{abstract}


\section{Introduction} \label{sec_intro}
\IEEEPARstart{W}{ith} the growth in wireless infrastructure and personal wireless devices, the demand for wireless ranging has mushroomed over the last decade. Wireless ranging involves finding the distance (also called range) between two wireless devices and it is a key step of most positioning/localization solutions \cite{Zafari2019}, and is also useful in other applications like proximity detection, wireless sensing, etc. The use cases are plentiful, including smart homes and buildings, context awareness, surveillance, disaster management, industry, health-care, etc. Correspondingly several different ranging techniques have been proposed that utilize Ultra-wide band (UWB) \cite{IEEE_154z}, Lidar \cite{Royo2019}, Global Navigation Satellite System (GNSS) \cite{Misra2011}, Ultrasound \cite{Gabbrielli2023}, Bluetooth \cite{Zand2019, Giovanni2021}, Zigbee \cite{IEEE_154} and Wi-Fi technologies \cite{IEEEWiFi2020, IEEE_11az}. Due to its low cost, wide indoor coverage and ubiquitous availability, Wi-Fi-based ranging is the ideal solution for many use cases and is the focus of this paper. 

Due to its attractiveness for many use cases, there is a rich literature of prior work on Wi-Fi-based ranging \cite{LiuFen2020}, an overview of which can be found later in this section. With the current Wi-Fi standardization efforts in IEEE 802.11mc \cite{IEEEWiFi2020} and 802.11az \cite{IEEE_11az}, round-trip-time (RTT) based ranging is the most popular approach for range estimation with Wi-Fi. In this approach, an initiator station (STA) and responder STA periodically perform a round-trip exchange of frames, as illustrated in Fig.~\ref{Fig_FTM_illustrate}. For the $p$-th frame exchange, the transmission times ($t_p^{(1)}, t_p^{(3)}$) and reception times ($t_p^{(2)}, t_p^{(4)}$) of the two frames are measured at the corresponding transmitting and receiving devices. With the exchange of these measurement times, the RTT and, correspondingly, the \emph{absolute} range between the initiator and responder can be estimated. For several applications such has gesture recognition, pose estimation, velocity estimation, direction finding, etc. \cite{Mollyn2023}, knowing the absolute range may not be as important as tracking the \emph{differential} range and \emph{relative} range. Here, differential range is the change in the absolute range of a device between two adjacent measurements, while relative range is its integration, i.e., the change in the absolute range of a device over a time window of interest. The absolute, differential and relative range are illustrated pictorially in Fig.~\ref{Fig_dRange_illustrate}. With such applications in mind, in this work we shall focus on the estimation of differential range.
\begin{figure}[!htb]
\centering
\includegraphics[width= 0.48\textwidth]{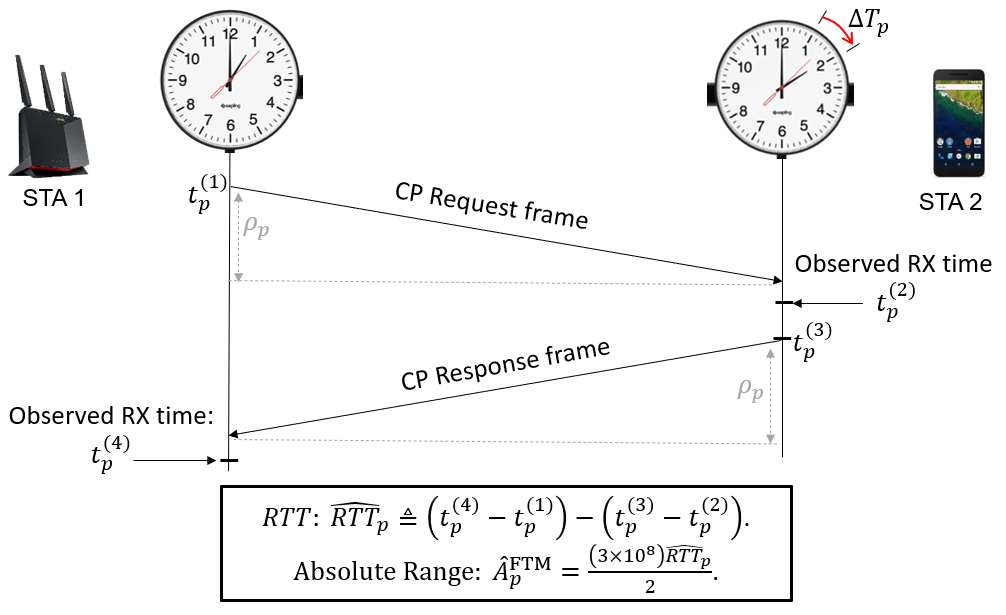}
\caption{An illustration of the $p$-th frame exchange between two Wi-Fi STAs for RTT-based range estimation. In EDCA Fine Time Measurement (FTM) protocol, the Request frame is an FTM frame and the Response frame is an ACK frame \cite{IEEE_11az}. In trigger-based and non-trigger-based FTM protocol, the Request frame is an I2R NDP frame and the Response frame is an R2I NDP frame \cite{IEEE_11az}. Here the clocks represent the fact that the time reference can be different at the two STAs.}
\label{Fig_FTM_illustrate}
\end{figure}
\begin{figure}[!htb]
\centering
\includegraphics[width= 0.48\textwidth]{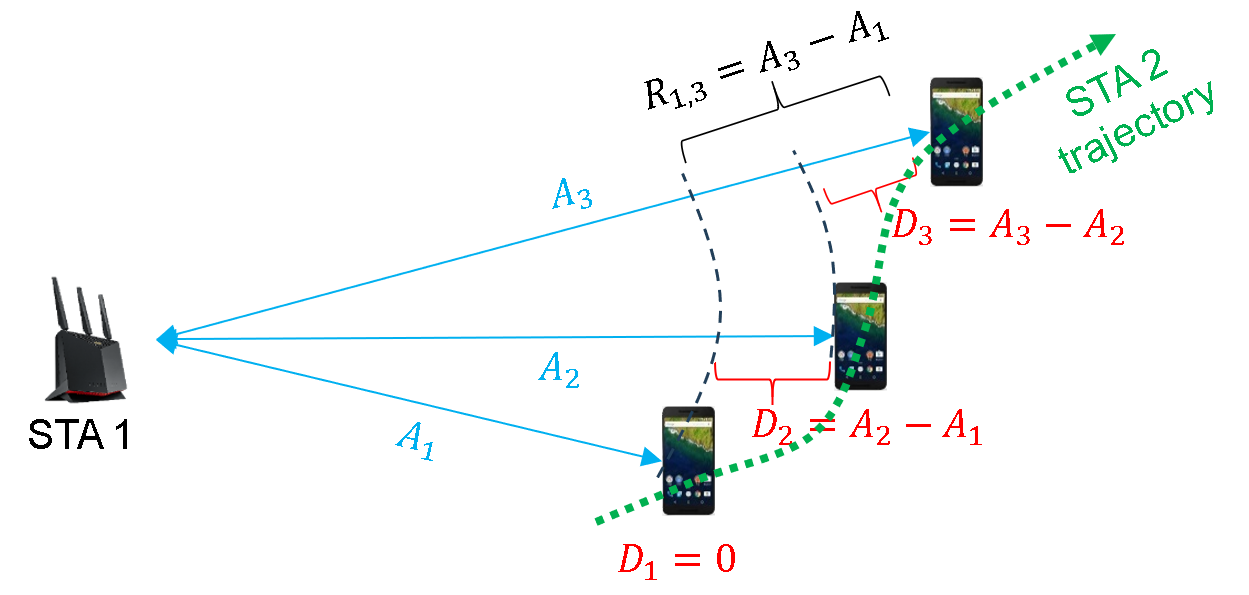}
\caption{An illustration of absolute range $A_p$, differential range $D_p \triangleq A_p-A_{p-1}$ and relative range $R_{q,p} \triangleq A_p-A_{q}$ frame exchanges $p \in \{1,2,3\}$ between two Wi-Fi STAs.}
\label{Fig_dRange_illustrate}
\end{figure}

One major limitation of conventional RTT-based range estimation is that the aforementioned time measurements only utilize the baseband signal information of the exchanged frames. This fundamentally limits the accuracy of conventional absolute/differential/relative range estimation by the system bandwidth (BW), as also justified by the baseband Cramer-Rao lower bound \cite{Thorbjornsen_2010}. 
Real-world measurements with the Fine Time Measurement (FTM) protocol have shown a ranging accuracy of only $\sim 1$ m even with an $80$ MHz system BW \cite{Ibrahim2018, Ma2022}, which may be unsuitable for several applications. However unlike absolute range, as shall be shown in this work, the differential range information can also be extracted from the passband signal information of the frames exchanged during the RTT ranging. Furthermore, this differential range estimate is not fundamentally limited by the system BW. Thus, by exploiting this passband information (which is conventionally untapped) a differential range estimate can be obtained that has a much higher precision than conventional estimates that use baseband information. 

In this work we propose WiDRa - a novel Wi-Fi-based Differential Ranging solution that exploits passband signal information and whose accuracy is not limited by the system BW. For this, WiDRa uses the carrier phase (CP) measurements from the frames exchanged within the Wi-Fi FTM protocol. Here CP is the phase of the received signal corresponding to the line-of-sight (LoS) path between the transmitter (TX) and receiver (RX), measured at the carrier frequency. 
The intuition is as follows: For a transmission by a TX, the phase of the received carrier at an RX with respect to its local oscillator depends on the propagation time – equivalently, range – between the TX and RX, the carrier frequency, and on impairments like initial phase offset and carrier frequency offset (CFO) between the TX and RX oscillators. If these impairments are compensated for, the remainder is proportional to the range between TX and RX and has a accuracy that is only limited by the carrier frequency. As shall be shown, our proposed method WiDRa can track the differential range with an accuracy of $< 1$ cm in a strong LoS scenario, even with just a $20$ MHz BW. Since the same frame exchanges yield both WiDRa-based and conventional RTT-based range estimates, they can also be combined to obtain higher precision differential, relative and absolute range estimates. 
In terms of prior work, the use of CP measurements for improving range estimation has been explored in GNSS systems with good success \cite{Remondi1985}. Fundamentally different from Wi-Fi, GNSS is a synchronized one-way ranging system where the measurements from multiple synchronized satellites, and measurements to two or more ground stations are used to deal with ambiguities in the CP measurements. CP-aided ranging has also been extended to other synchronized systems \cite{Yang2017, Dun2020}. 
More recently, the use of CP for improving ranging was also explored for 5G NR \cite{Chen2022}. However, there is no prior work, to the best of our knowledge, on the use of CP-based differential ranging for an unsynchronized system like Wi-Fi. As shall be shown, the lack of synchronization requires several innovations, such as combining the CP measurements from both the devices into a sum-CP metric and estimation and correction of the CFO. 
The contributions of this paper are as follows:
\begin{itemize}
\item We propose WiDRa - a new Wi-Fi CP-based method for tracking the differential range between a TX and RX that can reuse the frame exchanges from the FTM protocol. 
\item To validate the method, a detailed mathematical model for the Wi-Fi channel state information (CSI) is presented, taking into consideration the different system impairments that can impact CP-measurements.  \item We propose theoretically sound algorithms to correct these CSI impairments, obtain CP measurements from the corrected CSI, and recover the differential range from the CP measurements in LoS channels. In the process we also determine the necessary timing conditions to be satisfied by the frame exchanges. 
\item We test the proposed solutions on simulated data to study the impact of different system parameters on performance of WiDRA. As a proof-of-concept, the method is also successfully tested on a real-world test-bed.
\item We finally discuss the limitations of the current investigation and propose several future directions for work.
\end{itemize}
It should be noted that the focus of this paper is to introduce the idea of CP-based differential ranging, identify necessary conditions on frame transmissions, and show proof-of-concept (via theory, simulations and experiments) that the method works in the presence of real-world hardware impairments. Consequently, here we focus on LoS channels. Further investigations of the method in more challenging environments shall be explored in future work. 

\subsection{Prior Art on Wi-Fi ranging} \label{subsec_prior_art_intro}
\noindent Common approaches for Wi-Fi-based ranging between two devices are summarized below. 
\subsubsection{Model-based ranging} In this approach, a TX transmits a known signal, and an RX observes parameters of the received signal and uses a pre-determined model on how a transmitted signal is affected by the range to infer the range. The measured parameters are typically the Wi-Fi received signal strength indicator (RSSI) values, and the model can either be a pathloss equation \cite{Paul2009} or can be a data-driven model which is the case for finger-printing \cite{ Xia2017, Singh2021}. More recently, the use of CSI from one or more RX antennas, which provides much more information than RSSI, has also been explored as parameters have also been explored \cite{Yang2013, Rocamora2020, Kotaru2015}. This was historically the most common approach for ranging with Wi-Fi due to simplicity of the required parameters to be measured and lack of alternatives. 
However, these methods do not generalize well and require a significant overhead for calibrating the model to a specific deployment scenario, thus making them unappealing. 
\subsubsection{Phase difference of arrival-based ranging} In this approach, a TX transmits a narrow-band sinusoidal signal sequentially at multiple frequencies and the RX compares the relative phases of the received sinusoids to estimate range. This is a popular approach for ranging with synchronized narrow-band systems that are naturally designed for frequency hopping, such as Bluetooth \cite{Zand2019, Sheikh2023}. Some works have also extended its use to Wi-Fi systems by performing channel hopping \cite{Vasisht2016}. However, this method is not well suited for current Wi-Fi systems, since they are not designed to perform fast channel switching, the switching interrupts data transmission, and since the method doesn't exploit the wider BW available with Wi-Fi systems. 
With the introduction of multi-link operation in Wi-Fi 7 \cite{Evgeny2020}, there is some scope of exploitation of such methods without the need to switch channels. 
\subsubsection{Round trip time (RTT)-based ranging} In this approach, a TX and an RX complete a two-way frame exchange of frames, and use the time stamps of transmission and reception of both the forward and the backward  frames to estimate the RTT and, correspondingly, the range. Due to its lack of dependence on a prior model, training, or need for channel switching, this is currently the most popular approach for ranging with Wi-Fi. Correspondingly, there have been many prior works that have explored this method \cite{Li2000, Thorbjornsen_2010, Marcaletti2014, Kotaru2015, Rea2017}. The Fine Time Measurement (FTM) ranging protocol standardized by the 802.11mc \cite{IEEEWiFi2020} is also based on RTT-based ranging. The use of FTM measurements for passive ranging and the correction for the differing clock speeds at the TX and RX was explored in \cite{Banin2019}. The use of MUSIC algorithm to deal with multi-path was proposed in \cite{Kevin2020}. The performance of FTM for localization was studied in \cite{Ibrahim2018, Ma2022} and shown to provide an accuracy of $\sim 1$ m. Further security enhancements to the FTM protocol, and extensions to enable simultaneous multi-device ranging and passive ranging are being added in the latest 802.11az Wi-Fi standard \cite{IEEE_11az}. 

Methods to fuse the information from RSSI, CSI and/or RTT with each other and with information from other sensors have also been explored to improve the range estimates \cite{Woodman2008, Poulose2019, Yu2019, LiuXu2021, Guo2022}. 

The organization of the paper is as follows: the system model is discussed in Section \ref{sec_chan_model}; the estimation of CSI is analyzed in Section \ref{sec_csi_analysis}, and procedure to isolate the sum-CP from the CSI is described in Section \ref{sec_carrier_phase_est}; the differential range estimation using the sum-CP measurements is described in Section \ref{sec_dRange_est}; the necessary conditions on the frame transmissions are summarized in Section \ref{sec_summary_methodology}; the evaluations on simulated and real-field data are provided in Section \ref{sec_eval_results}; future directions of research are suggested in Section \ref{sec_future_dir}; and, finally, the conclusions are summarized in Section \ref{sec_conclusions}.

\textbf{Notation:} Scalars are represented by light-case letters; and sets by light-case calligraphic letters. Additionally, ${\rm j} = \sqrt{-1}$ and $\|a\|, \angle a, a^{*}$ represent magnitude, phase angle and the complex conjugate, respectively, of a complex scalar $a$. In addition, $\mathrm{mod}\{\cdot, a\}$ is the modulo function that provides the remainder after division with $a$, $\mathrm{Uni}[a,b]$ represents a uniformly distributed random variable in the range $[a,b]$, $\mathcal{CN}(a,b)$ represents a circularly-symmetric complex Gaussian random variable with mean $a$ and standard deviation $b$. Furthermore, $\mathrm{c} \triangleq 3 \times {10}^8$ m/s is the speed of light, $\mathrm{Re}\{\}$ represents the real component of a complex argument, $\mathbb{Z}$ is the set of integers, $\mathbb{R}$ is the set of real numbers, $\mathbb{C}$ is the set of complex numbers. 

\section{System model} \label{sec_chan_model}
\noindent We consider a system setup that has a stationary Wi-Fi access point as STA 1, a mobile Wi-Fi station as STA 2, both having a single antenna, as shown in Fig.~\ref{Fig_system_illustrate}. The goal is to estimate, at STA 1, the differential range of the STA 2 over a window of time.\footnote{Although the analysis is presented for range estimation being performed at STA 1, it can be readily extended for estimation at STA 2.} To enable the differential range estimation, STA 1 transmits a sequence of $P$ CP Request frames\footnote{Such frames can either be FTM frames in EDCA FTM protocol and I2R NDP frames in case of trigger based and non-trigger based FTM protocol.}, indexed as $\mathcal{P} = \{p \in \mathbb{Z} | 1 \leq p \leq P\}$. Let the transmit time of the $p$-th frame be $t_p^{(1)}$, as measured by STA 1. 
\begin{figure}[!htb]
\centering
\includegraphics[width= 0.48\textwidth]{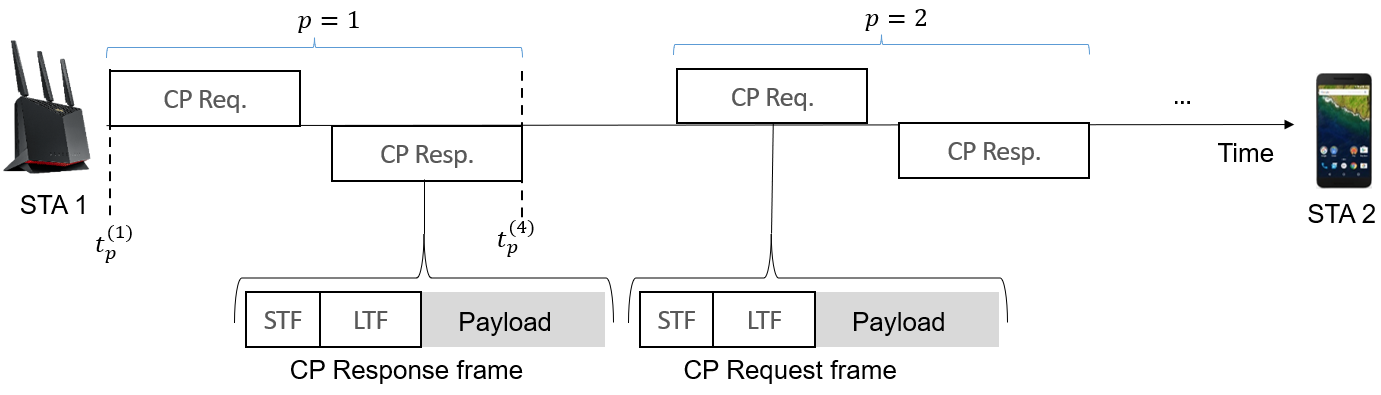}
\caption{An illustration of the CP Request and Response frame exchange between two stations to enable ranging.}
\label{Fig_system_illustrate}
\end{figure}
After receiving the $p$-th CP Request frame $p$, STA 2 also transmits the $p$-th CP Response frame 
and it is received by STA 1 at time $t_p^{(4)}$ (as measured by Station 1). To prevent any confusion, all the time values mentioned in this paper shall be in the time reference frame of STA 1. 
The header and the payload of these CP frames are encoded using orthogonal frequency division multiplexing (OFDM) over $K$ subcarriers indexed as $\mathcal{K} = \{k \in \mathbb{Z} | \lfloor \frac{-K+1}{2}\rfloor \leq k \leq \lfloor \frac{K-1}{2}\rfloor \}$ and with a symbol duration of $T_{\rm s}$ and cyclic prefix duration of $T_{\rm cy}$. From the Long Training Field of the $p$-th received CP Request frame and Response frame, STA 2 and STA 1 can estimate the CSI for each subcarrier $k$ as $\bar{h}_{p,k}^{(2)}$ and $\bar{h}_{p,k}^{(4)}$, respectively. Using these CSI measurements the STAs can further estimate the CP values $\widehat{\psi}^{(2)}_{p}$ and $\widehat{\psi}^{(4)}_{p}$ as defined later in Section \ref{subsec_symbol_timing_est}. 
We assume that the value of $\widehat{\psi}^{(2)}_{p}$ can be shared by STA 2 with STA 1 as the payload of a subsequent frame. 
In addition, we assume that using the training fields and pilot subcarriers of the CP Response frame $p$, STA 1 can also obtain an estimate of the carrier frequency offset $\bar{f}_{{\rm CFO},p}$ of STA 2 with respect to STA 1 with a precision of $\pm F$ Hz \cite{Sourour2004}.

For ease of analysis, we assume that the channel between the STAs is a purely LoS channel with a time-varying propagation delay $\rho_p$ at the time of $p$-th frame exchange. Note that the true range between the STAs is then $A_p \triangleq \mathrm{c} \rho_p$ (where $\mathrm{c}$ is the speed of light) and the differential range is $D_p \triangleq A_p-A_{p-1}$. The goal of the paper is to find accurate estimates of $\{D_p | 2 \leq p \leq P\}$ at STA 1 using the FTM time stamps $t_p^{(1)}$, $t_p^{(4)}$, the CFO $\bar{f}_{{\rm CFO},p}$ and the CP values $\widehat{\psi}^{(2)}_{p}$ and $\widehat{\psi}^{(4)}_{p}$ for $p \in \mathcal{P}$. The overall signaling and estimation steps performed at STA 1 and STA 2 are summarized in Fig.~\ref{Fig_summary_steps}. 
\begin{figure}[!htb]
\centering
\includegraphics[width= 0.48\textwidth]{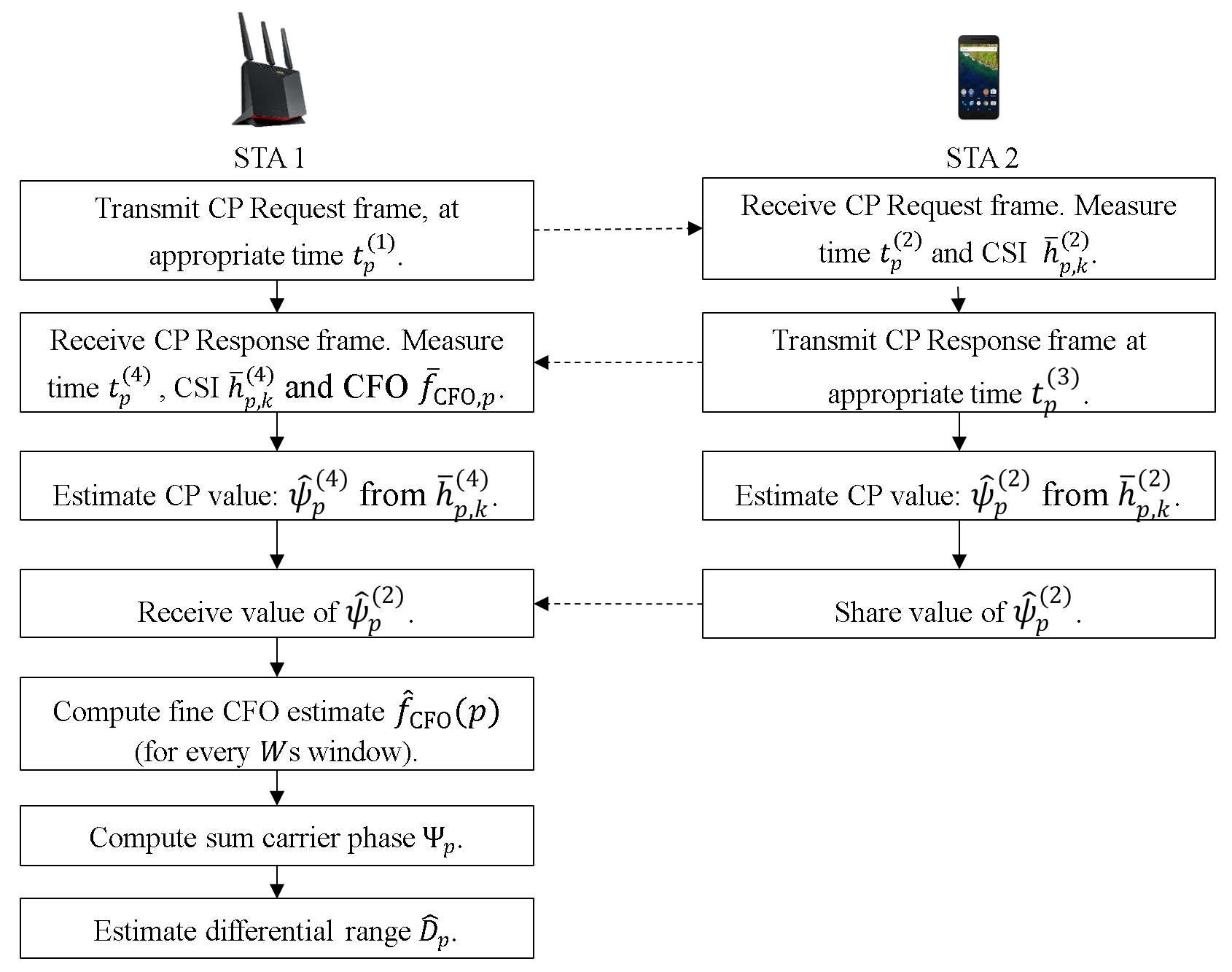}
\caption{Frame exchanges and steps performed by the STA 1 and STA 2.}
\label{Fig_summary_steps}
\end{figure}

\subsection{Modeling the non-idealities} \label{subsec_nonideal_model}
\noindent Without loss of generality, we assume that STA 1's crystal oscillator is accurate, while STA 2's crystal oscillator may differ from STA 1's oscillator by an unknown factor $\eta_p$ which can change slowly with time. For typical Wi-Fi systems, we have $|\eta_p| \leq 2 \times 10^{-5}$ \cite{IEEEWiFi2020} and $\eta_p$ has a coherence time of $\approx 1$ s. We also assume that all critical communication operations, including the carrier waveform, time stamps, and the analog-to-digital converter sampling clock at a STA are generated from the crystal oscillator and suffer the same offset \cite{Sourour2004}. Correspondingly, we have the following: 
\begin{enumerate}
\item The carrier frequency at STA 2 is $(1+\eta_p)f_{\rm c}$, where $f_{\rm c}$ is the carrier frequency at STA 1. Note that then the CFO estimate at STA 1 can be modeled as 
\begin{eqnarray} \label{eqn_CFO_est_init_model}
\bar{f}_{{\rm CFO},p} - \eta_p f_{\rm c} \sim \mathrm{Uni}[-F,F] . 
\end{eqnarray}
\item The sampling rate of the analog-to-digital converter at STA 2 is $K (1+\eta_p)/T_{\rm s}$.
\item The symbol duration for OFDM modulation/demodulation at STA 2 is $T_{\rm s}/(1 + \eta_p)$.
\item The timing clock at STA 2 runs faster by a factor of $(1+\eta_p)$ compared to STA 1.
\end{enumerate}
These can cause some phase noise \cite{Ratnam2019, Ratnam2020} and inter-carrier-interference in the CSI obtained by the STAs.
COTS Wi-Fi chips also suffer from several additional impairments which can impact the CSI estimated by them, as enumerated below. 
\begin{enumerate}
\item \emph{Gain error} - Due to the variable gain from the automatic gain control circuit, the amplitude of the CSI estimated from each received frame can vary significantly. This problem has been well investigated and several good gain error correction methods can be found in \cite{Ratnam2023}. Since the focus in this paper is on the phase of CSI, here we implicitly ignore the amplitude errors, assuming such corrections have already been performed.
\item \emph{Symbol timing error} - The symbol start time error arises primarily from errors in detecting the first arriving replica of the Short Training Field of a received frame. The error is dependent on the signal-to-noise ratio (SNR) and the BW of the system, and is independent for each received frame. Experiments have shown that these timing errors are usually smaller than $20 T_{\rm s}/K$ \cite{Ratnam2023}. 
\item \emph{Random phase rotation} - Experiments with several COTS Wi-Fi chips have shown that when receiving a frame, a STA may arbitrarily rotate the phase of the received packet by $2\pi n/N$ where $N$ is an implementation dependent parameter, and $0 \leq n < N$ is an arbitrary number (see Fig.~\ref{Fig_CFO_impact_illustrate}). Furthermore, the rotation can be different for each frame. This is likely due to cycle slips in phase-lock-loop circuits of the transceivers, and this behavior has also been observed by previous works in the context of angle-of-arrival estimation using inter-antenna phase \cite{Zhang2020, Zhuo2017, Zubow2021}. In this work we assume both STA1 and STA2 have the same value of $N$, and that $N$ is known. Typically we have $N=1, 2$ or $4$.
\end{enumerate}
Finally, we assume that due to differences in the front-end hardware of the uplink and downlink channels, the uplink channel from STA 2 to STA 1 (used for the CP Response frame) faces an additional fixed propagation delay of $\Delta\rho$.

\section{Analysis of CSI estimates} \label{sec_csi_analysis}
\noindent Given the system model, we have following lemmas on the CSI estimated from the CP Request and Response frames. 
\begin{lemma} \label{Lemma_CSI_acq}
If the CFO compensation during the OFDM demodulation at STA 2 is accurate, the CSI estimated from the $p$-th CP Request frame at STA 2 can be expressed as
\begin{eqnarray} \label{eqn_lemma_CSI_acq}
\bar{h}_{p,k}^{(2)} &=& e^{{\rm j} 2 \pi \frac{k}{T_{\rm s}} \tau_p^{(2)}} e^{{\rm j}2 \pi \frac{n^{(2)}_p}{N}} e^{-{\rm j} \phi_p} e^{- {\rm j} 2 \pi f_{\rm c} (1+\eta_p)\rho_p } ,
\end{eqnarray}
where $\tau_p^{(2)}$ is the symbol-start time detection error at STA 2, $2 \pi n^{(2)}_p \big/ N$ is the random phase rotation introduced by STA 2, and $\phi_p$ is the phase-offset of the carrier waveform of STA 2 with respect to the carrier waveform of STA 1 at time $t_p^{(1)}$.
\end{lemma}
\begin{proof}
See Appendix \ref{appdix_csi_acq_frame}.
\end{proof}
\begin{lemma} \label{Lemma_CSI_resp}
If the CFO compensation during the OFDM demodulation at STA 1 is accurate, the CSI estimated from the $p$-th CP Response frame at STA 1 can be expressed as
\begin{eqnarray} \label{eqn_lemma_CSI_resp}
\bar{h}_{p,k}^{(4)} &=& e^{{\rm j} 2 \pi \frac{k}{T_{\rm s}} \tau_p^{(4)}} e^{{\rm j}2 \pi \frac{n^{(4)}_p}{N}} e^{{\rm j} \phi_p} e^{- {\rm j} 2 \pi f_{\rm c} (1+\eta_p)(\rho_p+\Delta\rho) } \nonumber \\
&& \times e^{{\rm j} 2 \pi \eta_p f_{\rm c} [t_p^{(4)}-t_p^{(1)}]} ,
\end{eqnarray}
where $\tau_p^{(4)}$ is the symbol-start time detection error at STA 1, $2 \pi n^{(4)}_p \big/ N$ is the random phase rotation introduced by STA 1, $\Delta\rho$ is the additional propagation delay of the uplink channel (see Section \ref{subsec_nonideal_model}), and $\phi_p$ is as defined in Lemma \ref{Lemma_CSI_acq}.
\end{lemma}
\begin{proof}
See Appendix \ref{appdix_csi_resp_frame}.
\end{proof}
Intuitively in either of the above lemmas, we note that (i) the first term is the baseband phase rotation due to symbol start error, (ii) the second term is the random phase rotation introduced for implementation-specific reasons, (iii) the third term is the initial CP offset, (iv) the fourth term is the CP rotation due to the propagation delay, (v) the fifth term (if present) is the phase-rotation due to CFO from start of frame exchange till the frame detection time. 

\section{Isolating the sum carrier phase component} \label{sec_carrier_phase_est}
\noindent From \eqref{eqn_lemma_CSI_acq} and \eqref{eqn_lemma_CSI_resp}, note that we are interested in isolating the phase of fourth term, viz. $2 \pi f_{\rm c} (1+\eta_p) \rho_p$, which is the CP component. For doing so we need to get rid of as many of the other terms as possible, which is discussed in this section.

\subsection{Estimating CP values} \label{subsec_symbol_timing_est}
\noindent From \eqref{eqn_lemma_CSI_acq} and \eqref{eqn_lemma_CSI_resp}, we note that $\bar{h}_{p,k}^{(2)}, \bar{h}_{p,k}^{(4)}$ have a linear phase response with $k$, whose slope is determined by the symbol start time errors $\tau_p^{(2)}, \tau_p^{(4)}$. 
Therefore, their estimates $\widehat{\tau}_p^{(2)}, \widehat{\tau}_p^{(4)}$ can be obtained by running the root-MUSIC algorithm \cite{Rao1989} on $\bar{h}_{p,k}^{(2)}$ and $\bar{h}_{p,k}^{(4)}$ at STA 2 and STA 1, respectively, and picking the most dominant channel delay. 
Several other lower-complexity alternatives also exist to estimate $\tau_p^{(2)}, \tau_p^{(4)}$, but are skipped here for brevity. We can then estimate the CP values as
\begin{eqnarray} \label{eqn_psi_correct_music}
\widehat{\psi}^{(\bullet)}_{p} &=& \angle \left[ \sum_{k \in \mathcal{K}} \bar{h}_{p,k}^{(\bullet)} e^{-{\rm j} 2 \pi \frac{k}{T_{\rm s}} \widehat{\tau}_p^{(\bullet)}} \right],
\end{eqnarray}
where $\bullet = 2,4$ for STA 2 and STA 1 respectively. The CP value $\widehat{\psi}^{(2)}_{p}$ is then shared by STA 2 with STA 1 for further processing.

\subsection{Refining the CFO estimate} \label{subsec_CFO_est}
\noindent If the estimates $\widehat{\tau}_p^{(2)}$ and $\widehat{\tau}_p^{(4)}$ from Section \ref{subsec_symbol_timing_est} are accurate, we can show using \eqref{eqn_lemma_CSI_acq} and \eqref{eqn_lemma_CSI_resp}, that
\begin{subequations} \label{eqn_psi_simplify}
\begin{eqnarray}
\widehat{\psi}^{(2)}_{p} & \approx & {\rm mod}\bigg\{ 2 \pi \frac{n^{(2)}_p}{N} - \phi_p - 2 \pi f_{\rm c} (1+\eta_p)\rho_p, 2\pi \bigg\}, \\
\widehat{\psi}^{(4)}_{p} & \approx & {\rm mod}\bigg\{2 \pi \frac{n^{(4)}_p}{N} + \phi_p - 2 \pi f_{\rm c} (1+\eta_p)(\rho_p+\Delta\rho) \nonumber \\
&& + 2 \pi \eta_p f_{\rm c} [t_p^{(4)}-t_p^{(1)}], 2\pi \bigg\}.
\end{eqnarray}
\end{subequations}
As is evident, one of the remaining unwanted terms is the one dependent on the CFO $\eta_p f_{\rm c}$. Although an estimate of the CFO is already available at STA 1 from the response frame (see \eqref{eqn_CFO_est_init_model}), its precision of $\pm F$ Hz may be insufficient for CP-based ranging. Correspondingly, in this subsection we propose methods to refine it. 
From \eqref{eqn_psi_simplify} note that we can write
\begin{subequations} \label{eqn_psi_simplify_diff}
\begin{flalign}
& N \big(\widehat{\psi}^{(2)}_{p} - \widehat{\psi}^{(2)}_{p-1} \big) \stackrel{2 \pi}{\approx} N(\phi_{p-1} - \phi_p) - 2 \pi N f_{\rm c} (\rho_p-\rho_{p-1}), & \\
& N \big( \widehat{\psi}^{(4)}_{p} - \widehat{\psi}^{(4)}_{p-1} \big) \stackrel{2 \pi}{\approx} N(\phi_p - \phi_{p-1}) - 2 \pi N f_{\rm c} (\rho_p - \rho_{p-1}) & \nonumber \\
& \qquad \qquad + 2 \pi N \eta_p f_{\rm c} [t_p^{(4)}-t_p^{(1)} - t_{p-1}^{(4)} + t_{p-1}^{(1)}], &
\end{flalign}
\end{subequations}
where we define $\stackrel{2 \pi}{\approx}$ to mean approximate equivalence modulo $2\pi$ and use the facts that $n^{(2)}_p$ and $n^{(4)}_p$ are integers, and $\eta_p f_{\rm c} (\rho_p - \rho_{p-1})\approx 0$. Since $\phi_p$ is defined as the CP offset at the beginning of frame exchange $p$, it further follows that $\phi_p = \phi_{p-1} + \pi (\eta_{p}+\eta_{p-1}) f_{\rm c} [t_p^{(1)} - t_{p-1}^{(1)}]$. Applying this to \eqref{eqn_psi_simplify_diff} we get
\begin{flalign} \label{eqn_psi_simplify_diff_2}
& N \big(\widehat{\psi}^{(4)}_{p} - \widehat{\psi}^{(2)}_{p} - \widehat{\psi}^{(4)}_{p-1} + \widehat{\psi}^{(2)}_{p-1} \big) \stackrel{2\pi}{\approx} & \nonumber \\
& \qquad \qquad  2 \pi N f_{\rm c}\eta_p [t_p^{(4)} - t_{p-1}^{(4)} + t_p^{(1)} - t_{p-1}^{(1)}] , 
\end{flalign}
where we use the fact that $\eta_p \approx \eta_{p-1}$ since crystal offset drifts slowly. Note that \eqref{eqn_psi_simplify_diff_2} is only dependent on the CFO $\eta_p f_{\rm c}$ and is independent of the other unknowns. This analysis is justified using data from an example real-world measurement in Fig.~\ref{Fig_CFO_impact_illustrate}, where we clearly observe a phase ramp caused by the CFO. We also observe two ramps shifted by $\pi$ radians showing the random phase rotation ($N=2$). 
\begin{figure}[!htb]
\centering
\includegraphics[width= 0.4\textwidth]{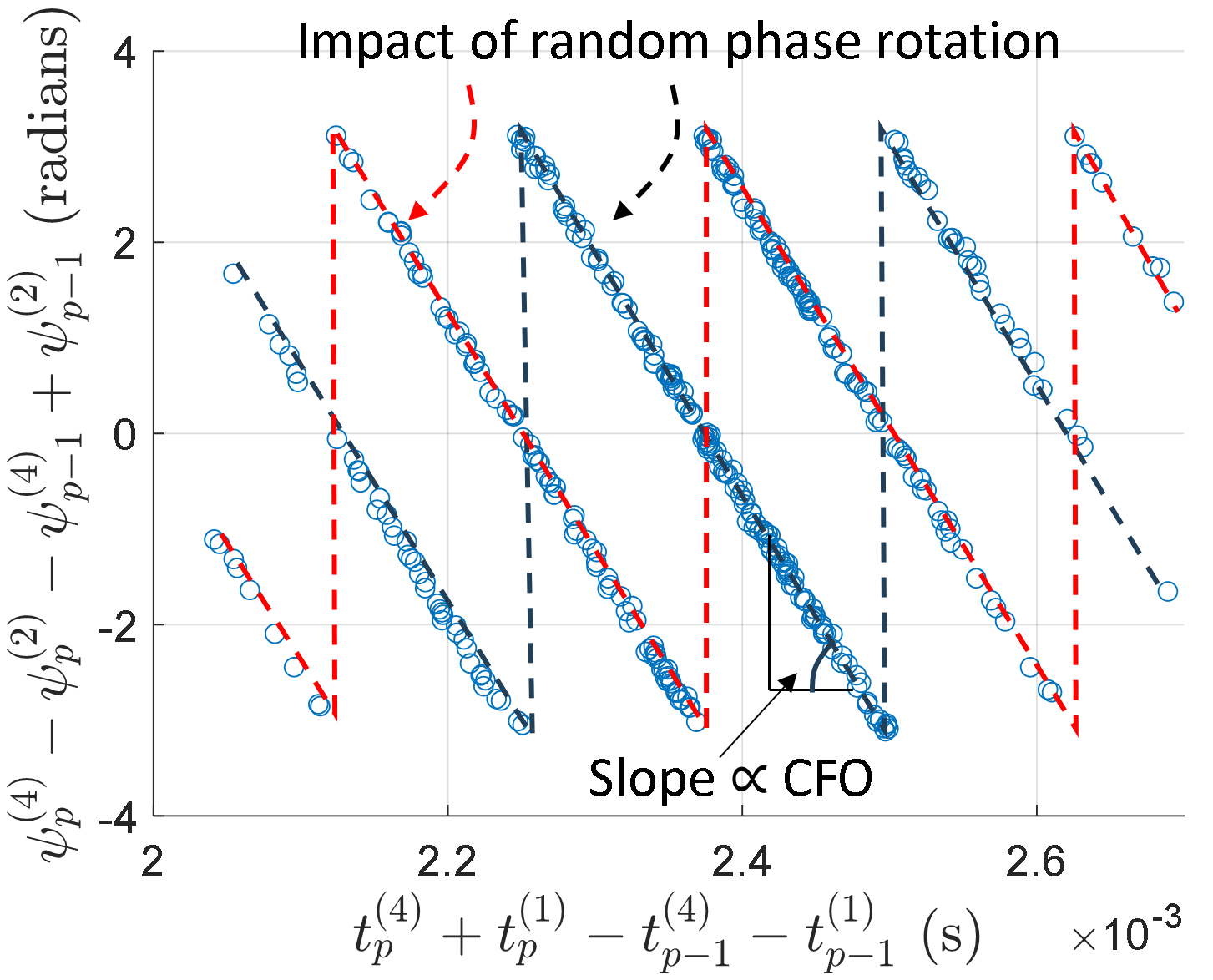}
\caption{A scatter plot of $\widehat{\psi}^{(4)}_{p} - \widehat{\psi}^{(2)}_{p} - \widehat{\psi}^{(4)}_{p-1} + \widehat{\psi}^{(2)}_{p-1}$ versus $t_p^{(4)} + t_p^{(2)} - t_{p-1}^{(4)} - t_{p-1}^{(2)}$ from a measurement campaign involving two Intel AX210 chips, with other parameters as in Section \ref{subsec_real_data}. The data belongs to a time window of $W=0.5$ s with the STA 2 speed being $0.35$ m/s.}
\label{Fig_CFO_impact_illustrate}
\end{figure}
Although $t_p^{(1)} - t_{p-1}^{(1)}$ and $t_p^{(4)} - t_{p-1}^{(4)}$ can be larger than one cycle of the CFO, we can leverage two facts to find a unique estimate for the CFO $\widehat{f}_{\rm CFO}(p)$ from \eqref{eqn_psi_simplify_diff_2}, viz., (i) we have that $|\eta_p f_{\rm c} - \bar{f}_{{\rm CFO},p}| \leq F$ and, (ii) $\eta_p$ remains approximately the same within each window of $W=0.5$ s. 
In our proposed approach, a separate value of CFO is estimated from the samples $\widehat{\psi}^{(\bullet)}_{p}$ within each time window of $W=0.5$ s and an interpolation is applied, as summarized in Algorithm \ref{Algo2}. 
\begin{algorithm}
\label{Algo2}
\caption{Estimation of CFO $\widehat{f}_{\rm CFO}(p)$}
\begin{algorithmic} 
\STATE Inputs: $\widehat{\psi}^{(2)}_{p}, \widehat{\psi}^{(4)}_{p}, t_p^{(1)}, t_p^{(4)}, \bar{f}_{{\rm CFO},p}$ for $p \in \mathcal{P}$.
\STATE $W = 0.5$ s.
\STATE $T_{\rm max} = 1$ ms. //$T_{\rm max}$ depends on the crystal stability and is set such that for any $p,q \in \mathcal{P}$ s.t. $|t_p^{(1)}-t_q^{(1)}| < W$, we have $(\eta_p-\eta_q) f_{\rm c} T_{\rm max} \ll 1$. 
\STATE $W_{\rm max} = \lceil (t_P^{(1)}-t_1^{(1)})/W \rceil$.
\STATE $\bar{\mathcal{Q}} = \{\}$.
\FOR{$w = 1:1:W_{\rm max}$}
\STATE $t_{\rm st} = t_1^{(1)} + (w-1) W$.
\STATE $t_{\rm end} = t_1^{(1)} + w W$.
\STATE $\mathcal{Q} = \big\{p \in \mathcal{P} \big| t_{\rm st} \leq t_p^{(1)} \leq t_{\rm end}, |t_p^{(1)}-t_{p-1}^{(1)}| < T_{\rm max} \big\}$.
\STATE $\bar{\bar{f}}_{\rm CFO} = \sum_{p \in \mathcal{Q}} \bar{f}_{{\rm CFO}, p} / |\mathcal{Q}|$.
\STATE Compute 
\begin{flalign}
& f_{\rm opt} = \argmax_{|f| \leq F} \mathrm{Re} \bigg\{ \sum_{p \in \mathcal{Q}}  e^{{\rm j} N \big(\widehat{\psi}^{(4)}_{p} - \widehat{\psi}^{(2)}_{p} - \widehat{\psi}^{(4)}_{p-1} + \widehat{\psi}^{(2)}_{p-1}\big)} & \nonumber \\
& \qquad \quad \times e^{ -{\rm j} 2\pi N (f + \bar{\bar{f}}_{\rm CFO}) [t_p^{(4)} + t_p^{(1)} - t_{p-1}^{(4)} - t_{p-1}^{(1)}]} \bigg\} &
\end{flalign}
\STATE $\widehat{f}_{\rm CFO}(p) = f_{\rm opt} + \bar{\bar{f}}_{\rm CFO}$ for $p \in \mathcal{Q}$.
\STATE $\bar{\mathcal{Q}} = \bar{\mathcal{Q}} \cup \mathcal{Q}$.
\ENDFOR
\FOR{$p \in \mathcal{P} \setminus \bar{\mathcal{Q}}$}
\STATE Set $\widehat{f}_{\rm CFO}(p)$ by linear-interpolation of $\widehat{f}_{\rm CFO}(\cdot)$ values from $\bar{\mathcal{Q}}$.
\ENDFOR
\STATE Return $\big\{ \widehat{f}_{\rm CFO}(p) | p \in \mathcal{P} \big\}$.
\end{algorithmic}
\end{algorithm}
We also have the following remark:
\begin{remark} \label{remark_jitter_CFO}
To estimate the CFO accurately, Algorithm \ref{Algo2} requires that each time window of $W$ s has at least $3$ samples with the inter-frame exchange spacing $t_p^{(1)}-t_{p-1}^{(1)} < T_{\rm max}$. Furthermore, for these samples the time interval $t_p^{(4)}+t_{p}^{(1)} - t_{p-1}^{(4)} -t_{p-1}^{(1)}$ must (i) either be smaller than $1/(4 N F)$ or (ii) have a jitter of the order of $1/(4 N F)$. 
\end{remark}

\subsection{Estimating the sum-CP component}
\noindent Adding the two phase values in \eqref{eqn_psi_simplify}, multiplying by $N$ and subtracting out the CFO term, we define the the sum-CP as 
\begin{eqnarray}
{\Psi}_{p} & \triangleq & \mathrm{mod} \big\{ N (\widehat{\psi}^{(2)}_{p} + \widehat{\psi}^{(4)}_{p}) \nonumber \\
&& - 2 \pi N \widehat{f}_{\rm CFO}(p) [t_p^{(4)}-t_p^{(1)}] + \pi , 2\pi \big\} - \pi \label{eqn_phi_sum} \\
& \approx & {\rm mod}\big\{ - 2 N \pi f_{\rm c} (1+\eta_p) (2\rho_p + \Delta \rho) \nonumber \\
&& + 2 N \pi (\eta_p f_{\rm c} - \widehat{f}_{\rm CFO}(p)) [t_p^{(4)}-t_p^{(1)}] + \pi, 2\pi \big\} - \pi. \nonumber
\end{eqnarray}
As is evident, all the phase-impairments are now removed and, if the CFO estimate is accurate, ${\Psi}_{p}$ can track changes in the range accurately.
However the CFO removal in \eqref{eqn_phi_sum} only performs a first-order linear fit for the phase offset between STA 1 and STA 2, and any residual phase noise may still impact ${\Psi}_{p}$. This is exemplified in Fig.~\ref{Fig_PhN_impact_illustrate}, where we plot real-world measurements of ${\Psi}_{p}$ for varying values of $t_p^{(4)}-t_p^{(1)}$ for a stationary STA 2. As evident, the mean value of ${\Psi}_{p}$ stays approximately constant over time for a stationary STA 2 (which is desired), but the residual phase noise causes the variance to grow with $t_p^{(4)}-t_p^{(1)}$. We correspondingly have the following remark:
\begin{remark} \label{remark_SIFS_spacing}
To keep the impact of any residual phase-noise in \eqref{eqn_phi_sum} low, it is desirable to design the frame exchange sequence such that $t_p^{(4)}-t_p^{(1)}$ is smaller than $1$ms. 
\end{remark}
\begin{figure}[!htb]
\centering
\includegraphics[width= 0.4\textwidth]{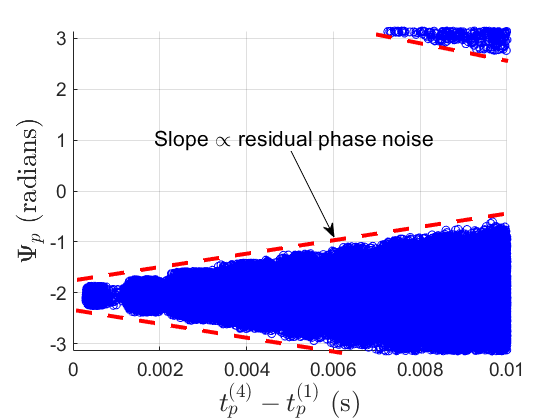}
\caption{A scatter plot of ${\Psi}_{p}$ versus $t_p^{(4)} - t_{p}^{(1)}$ from a measurement campaign involving two Intel AX210 chips, with other parameters as in Section \ref{subsec_real_data}. The measurements span a time period of $30$ s with both STA 1 and STA 2 being static.}
\label{Fig_PhN_impact_illustrate}
\end{figure}

\section{Differential and relative range estimation} \label{sec_dRange_est}
\noindent From \eqref{eqn_phi_sum}, we see that ${\Psi}_{p}$ is linearly related with the propagation delay $\rho_p$ and, since it tracks the product $2 f_{\rm c} \rho_p$, it can track range changes that are of the order of half a wavelength ($3$ cm at $5$ GHz). This is a $100 \times$ better resolution than RTT-based range estimation. However, the presence of the ${\rm mod}\{., 2 \pi\}$ in \eqref{eqn_phi_sum} introduces an ambiguity in the range estimation. In other words, a propagation delay of $\rho$ and $\rho + 1/(2 N (1+\eta_p) f_{\rm c})$ yield the same value of the corrected sum-CP ${\Psi}_{p}$. This is somewhat similar to the integer ambiguity problem faced in many phase-based ranging methods \cite{Vasisht2016}. Rather than trying to resolve this integer ambiguity by performing frequency hopping, here we use the sum-CP for differential ranging. Note that if the change in the propagation delay $\rho_p$ between adjacent frame exchanges satisfies $|\rho_p - \rho_{p-1}| < 1/(4 N (1+\eta_p) f_{\rm c})$, from \eqref{eqn_phi_sum} we can write
\begin{eqnarray} \label{eqn_Phi_diff}
\mathrm{mod}\{ {\Psi}_{p} - {\Psi}_{p-1} + \pi, 2\pi \} - \pi \approx - 4 \pi N f_{\rm c} (\rho_p - \rho_{p-1}),
\end{eqnarray}
where we use the fact that $\eta_p f_{\rm c} |\rho_p - \rho_{p-1}| \ll 1$, and we are able to get rid of the ambiguity. In other words, the differential range for the $p$-th frame exchange can be estimated without ambiguity as
\begin{eqnarray} \label{eqn_dRange_est}
\widehat{D}_p = - \mathrm{c} \times \frac{\mathrm{mod}\{ {\Psi}_{p} - {\Psi}_{p-1} + \pi, 2\pi \} - \pi}{4 \pi N f_{\rm c}} ,
\end{eqnarray}
as long as STA 2 movement between adjacent measurements is less than $1/4N$-th of a carrier wavelength. Thus, have the following observation on the maximum STA 2 speed:
\begin{remark} \label{remark_max_speed}
The maximum STA 2 speed that can be unambiguously tracked using \eqref{eqn_dRange_est} is 
\begin{eqnarray} \label{eqn_max_user_speed}
S_{\rm max} = {\rm c} \big/ \big[ 4 N f_{\rm c} \times \max_p\{t_p^{(1)}-t_{p-1}^{(1)}\} \big] .
\end{eqnarray}
\end{remark}
Note that Wi-Fi channels are susceptible to occupancy by other devices, which can cause $t_p^{(1)}-t_{p-1}^{(1)}$ to be large for some outlier $p$ and cause missing of $1/2N$-wavelength cycles in \eqref{eqn_dRange_est}. We refer to this effect as \emph{cycle slips}. 
Estimation that is robust to such outliers is possible by exploiting the temporal correlation of STA 2 speed, to estimate the range as
\begin{eqnarray} \label{eqn_dRange_est_2}
\widehat{D}_p = - \mathrm{c} \times \frac{\mathrm{mod}\{ {\Psi}_{p} - {\Psi}_{p-1} + \Theta_p + \pi, 2\pi \} - \Theta_p - \pi}{4 \pi N f_{\rm c}} , 
\end{eqnarray}
where
\begin{eqnarray}
\Theta_p &=& \sum_{q \in \mathcal{Q}_p} 4 \pi N f_{\rm c} \widehat{S}_p ( t^{(1)}_p - t^{(1)}_{q}) / {\rm c}, \nonumber \\
\widehat{S}_p &=& \argmax_{|s| \leq S_{\rm max}} \Big| \sum_{q \in \mathcal{Q}_p} e^{{\rm j} (\Psi_q-\Psi_p + \frac{4 \pi N f_{\rm c}}{\rm c} (t^{(1)}_q-t^{(1)}_p) s)}\Big|, \nonumber \\
\mathcal{Q}_p &=& \big\{ q \in \mathcal{P} \big| |t_q^{(1)} - t_p^{(1)}| \leq 0.25 \big\}, \nonumber \\
S_{\rm max} &=& {\rm c} \big/ \big[4Nf_{\rm c} \times {\rm percentile}_{80}\{ t_p - t_{p-1} | p \in \mathcal{P}\} \big] . \nonumber
\end{eqnarray}
Both \eqref{eqn_dRange_est} and \eqref{eqn_dRange_est_2} are \emph{one-shot} estimates that do not use temporal smoothing. While outside the scope of this paper, appropriate state-space models for differential range may be applied to further improve performance via Kalman filtering \cite{Welch1995}. 

The differential range estimates can further be integrated to obtain the \emph{relative range} estimates as
\begin{eqnarray} \label{eqn_rel_range}
\widehat{R}_{q,p} \triangleq \sum_{a=q+1}^{p} \widehat{D}_a, 
\end{eqnarray}
which measures the change in the range between time stamps $t_q^{(1)}$ and $t_p^{(1)}$ ($p > q$). As the user speed reaches closer to $S_{\rm max}$, we may experience cycle-slips which lead to accumulative error in $\widehat{R}_{q,p}$ for large $t_p^{(1)} - t_q^{(1)}$. This affect is studied later in Section \ref{subsec_sim_data}.

\section{Restrictions on transmission times} \label{sec_summary_methodology}
\noindent Based on the presented analysis above, we have the following recommendations on the transmit times of CP request/response frames:
\begin{enumerate}
\item To keep the frame-exchange overhead low, the median inter-frame exchange spacing $t_p^{(1)}-t_{p-1}^{(1)}$ should be $\geq 25$ ms.
\item For CFO estimation (see Remark \ref{remark_jitter_CFO}), we recommend adding a few additional frame exchanges every $W$ s with the inter-frame exchange spacing $t_p^{(1)}-t_{p-1}^{(1)}$ being small $\approx 300 \ \mu$s with a jitter of $1/(4NF)$.\footnote{These frames should ideally be exchanged within a single transmit opportunity to avoid channel contention delays. Jitter, if required, can be provided via zero padding \cite{IEEEWiFi2020}.} Such a \emph{nested}-structure is depicted in Fig.\ref{Fig_nested_array_jitter}. 
\item Based on Remark \ref{remark_SIFS_spacing}, we recommend using only a Short Inter-Frame Spacing ($16 \  \mu$s) between the end of reception of CP Request frame and transmission time of CP Response frame at STA 2. 
\end{enumerate}
\begin{figure}[!htb]
\centering
\includegraphics[width= 0.48\textwidth]{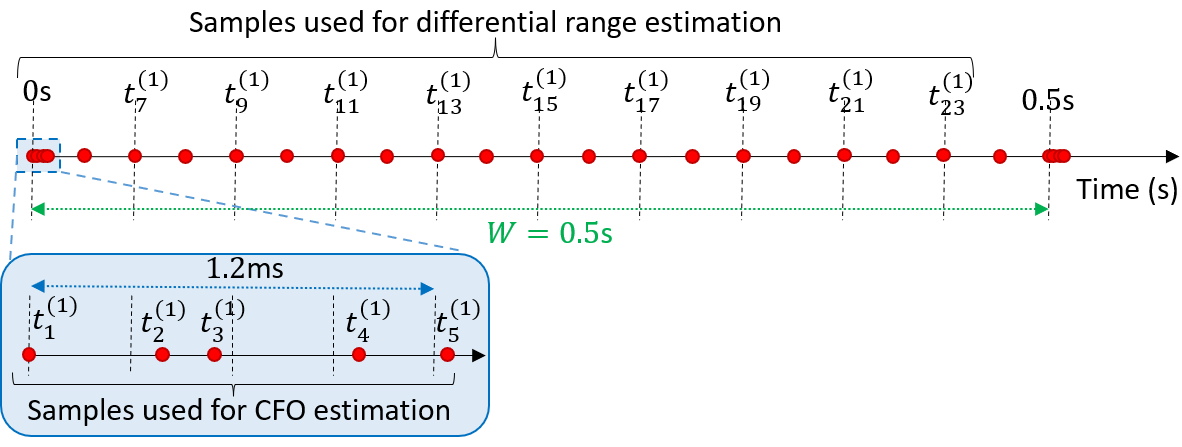}
\caption{An illustration of the jittered nested array structure for the frame exchange transmission times $\{t^{(1)}_p | p \in \mathcal{P}\}$ with outer-array spacing of $25$ ms and inner array spacing of $\approx 300 \ \mu$s.}
\label{Fig_nested_array_jitter}
\end{figure}

\section{Evaluation Results} \label{sec_eval_results}
\noindent For the evaluating the impact of different system parameters, we first use simulated CSI data. Later, we shall also validate the proposed method on real-world CSI data to show proof-of-concept. As a baseline for comparison, we use the differential range estimate $\widehat{D}_p^{\rm FTM} = \widehat{A}^{\rm FTM}_p - \widehat{A}^{\rm FTM}_{p-1}$, where $\widehat{A}^{\rm FTM}_p$ is the conventional absolute range estimate from the FTM procedure that uses the time measurements and baseband information. Here $\widehat{A}^{\rm FTM}_p$ is given by 
\begin{eqnarray} \label{eqn_FTM_Arange_est}
\widehat{A}_p^{\rm FTM} = \frac{\mathrm{c}}{2} \left( t^{(4)}_p - \widehat{\tau}^{(4)}_p - t^{(1)}_p - \frac{t^{(3)}_p - t^{(2)}_p + \widehat{\tau}^{(2)}_p}{1 + \bar{f}_{{\rm CFO},p}/f_{\rm c}} \right),
\end{eqnarray}
where $t^{(2)}_p, t^{(3)}_p$ are the time of reception of the CP request frame and transmission of CP response frame as measured by STA 2 (see Fig.~\ref{Fig_FTM_illustrate}), and $\widehat{\tau}^{(2)}_p, \widehat{\tau}^{(4)}_p$ are the estimates from Section \ref{subsec_symbol_timing_est}. For both WiDRa and FTM, we use root-MUSIC with signal-dimension set to $3$ to estimate $\widehat{\tau}^{(2)}_p, \widehat{\tau}^{(4)}_p$. Note that \eqref{eqn_FTM_Arange_est} includes the state-of-the-art FTM corrections for the symbol start time and differing clock speed \cite{Banin2019, Kevin2020}. Although we compare WiDRa to FTM here, it should be emphasized that the CP-based method uses the passband information while FTM-based ranging uses the baseband information from the same frame exchanges. Thus, they are complementary to each other and can be combined to improve overall range estimation as discussed in Section \ref{sec_future_dir}. 

\subsection{Simulated Data} \label{subsec_sim_data}
\noindent For simulations, we consider a 2-dimensional room of dimension $5 \times 5$ m with STA 1 located at its center and consider a simulation time of $30$ s. In each random epoch, STA 2 moves with a uniform speed of $S$ m/s using the Random waypoint model \cite{Hyytia2007} as depicted in Fig.~\ref{Fig_sim_layout}. 
STA 1 and STA 2 operate on a $20$ MHz channel at $f_{\rm c} = 5.2$ GHz (channel 40 of the $5$ GHz ISM band) using the 802.11ax protocol \cite{IEEEWiFi_11ax}. The transmission is via OFDM modulation with a symbol duration $T_{\rm s} = 12.8 \ \mu$s, cyclic prefix duration $T_{\rm cy} = 3.2 \ \mu$s and $K = 256$ sub-carriers. 
For enabling the differential ranging, STA 1 transmits $P=1440$ CP Request frames using a jittered nested array structure for transmit times as
\begin{eqnarray}
t_{p}^{(1)} = 0.5 \lfloor (p-1)/24 \rfloor + 0.025 ({\rm mod}\{p-1, 24\}-4), \nonumber
\end{eqnarray}
for ${\rm mod}\{p-1, 24\} \geq 5$ and,
\begin{eqnarray}
t_{p}^{(1)} = 0.5 \lfloor (p-1)/24 \rfloor + 0.0003 {\rm mod}\{p-1, 24\} + \chi_p, \nonumber
\end{eqnarray}
for ${\rm mod}\{p-1, 24\} \leq 4$ with $\chi_p \sim \mathrm{Uni}[-100,100] \ \mu$s in each epoch. Note that this generates one frame exchange every $20$ms on average. The CP Request and Response frames are assumed to have transmission duration of $100 \ \mu$s each, and STA 2 transmits the CP Response frames $16 \ \mu$s after receiving the Request frame. 
\begin{figure}[!htb]
\centering
\includegraphics[width= 0.35\textwidth]{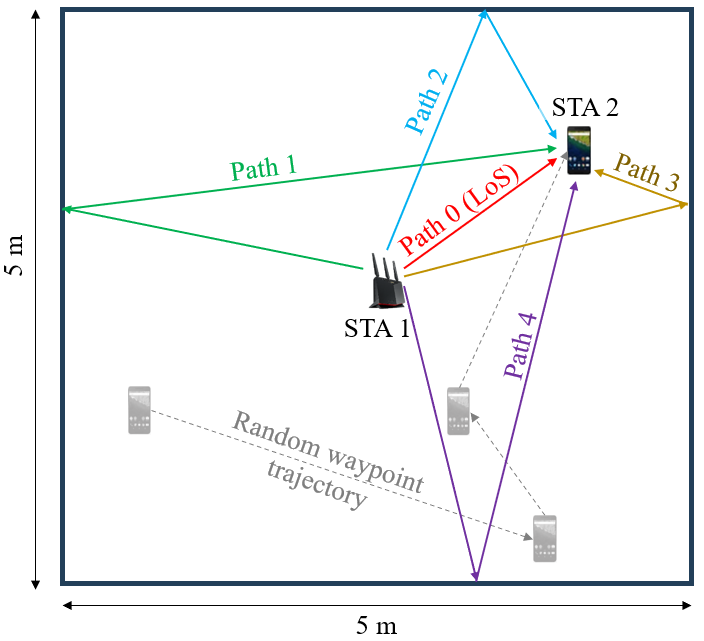}
\caption{An illustration of the simulation layout.}
\label{Fig_sim_layout}
\end{figure}

For generating the channel, we consider a simple abstraction of ray-tracing where the channel has 5 paths, viz., the LoS path, and a path each reflected from each of the $4$ walls of room as shown in Fig.~\ref{Fig_sim_layout}.\footnote{The results were found to be very similar if more reflections or scattering objects are considered while keeping the Rician factor the same.} The channel from STA 1 to STA 2 is assumed to have a Rician factor of $\kappa$, and is written as
\begin{eqnarray}
h_p(t) = \sqrt{\kappa} \frac{\delta(t-\rho_p)}{\sqrt{\kappa+1}} + \sum_{\ell=1}^4 \frac{\delta(t-\widetilde{\rho}_{p,\ell})}{\sqrt{4(\kappa+1)}},
\end{eqnarray}
where $\widetilde{\rho}_{p,\ell}$ is the propagation delay of the path reflected from the $\ell$-th wall. Note that these paths cannot be resolved with the $20$ MHz BW and do cause interference to the phase estimate in \eqref{eqn_psi_correct_music}. For the channel from STA 2 to STA 1, to model the uplink-downlink asymmetry, we consider the additional path length of $\Delta \rho = 1$ ns. 

The crystal oscillator at STA 1 is modeled to be perfect, while the one at STA 2 is modeled to have a sinusoidal offset given by
$$\eta_p = [0.5 + 0.1\sin(0.05\pi t_{p}^{(1)} + \varphi)] \times {10}^{-5},$$
where $\varphi \sim \mathrm{Uni[0,2\pi)}$ in each epoch. This $\eta_p$ is then used to determine the value of carrier frequency, sampling timing and symbol duration at STA 2. The OFDM modulation and OFDM demodulation are performed at STA1 and STA 2 as per \eqref{eqn_tx_sig_s}, \eqref{eqn_demod_Y_k}, \eqref{eqn_tx_sig_s_alt}, and \eqref{eqn_demod_Y_k_alt}. We assume a conservative receiver CFO estimation precision is $F=5$ KHz. Based on our prior experiments with COTS hardware \cite{Ratnam2023}, the symbol start time errors are modeled as $\tau^{(2)}_p, \tau^{(4)}_p \sim \mathrm{Uni}[-300,-100]$ ns. We also assume the STAs have the arbitrary phase rotation with $N=2$. For convenience, we do not model CSI gain errors, but the measured CSI is assumed to suffer from additive Gaussian noise. Finally, for the benefit of FTM, we assume that the measured time stamps have a good precision of $0.01$ ns.\footnote{Unlike FTM, WiDRa is resilient to time stamps and can do well even with a precision of $>10$ ns.}

For these settings, the impact of different system parameters on the RMSE of differential range estimation with WiDRa is depicted in Fig.~\ref{Fig_simulation_results}. Fig.~\ref{Fig_RMSE_vs_SNR} studies the impact of the SNR on performance, and as can be seen, the proposed method is robust to channel noise and can perform well even at a low SNR of $10$ dB. Fig.~\ref{Fig_RMSE_vs_Speed} investigates the impact of STA 2 speed on performance, and we observe an abrupt degradation of performance at $0.25$ m/s, but the RMSE is low below this speed. This is in good alignment with the analysis in Section \ref{sec_dRange_est}, since for chosen simulation parameters we have $S_{\rm max} = 0.28$ m/s from \eqref{eqn_max_user_speed}. Although the analysis in this paper focused on a single-path channel, in Fig.~\ref{Fig_RMSE_vs_Kfactor} we study the impact of channel multi-path on performance by varying the Rician factor. The results suggest that even in the presence of multi-path WiDRa performs well, if the channel Rician factor is $\kappa \geq 7$. Overall, we also observe that on an average WiDRa has a $100 \times$ lower RMSE for differential range estimation than the baseline one-shot FTM procedure (which uses RTT). 
\begin{figure}[!htb]
\centering
\subfloat[SNR]{\includegraphics[width= 0.4\textwidth]{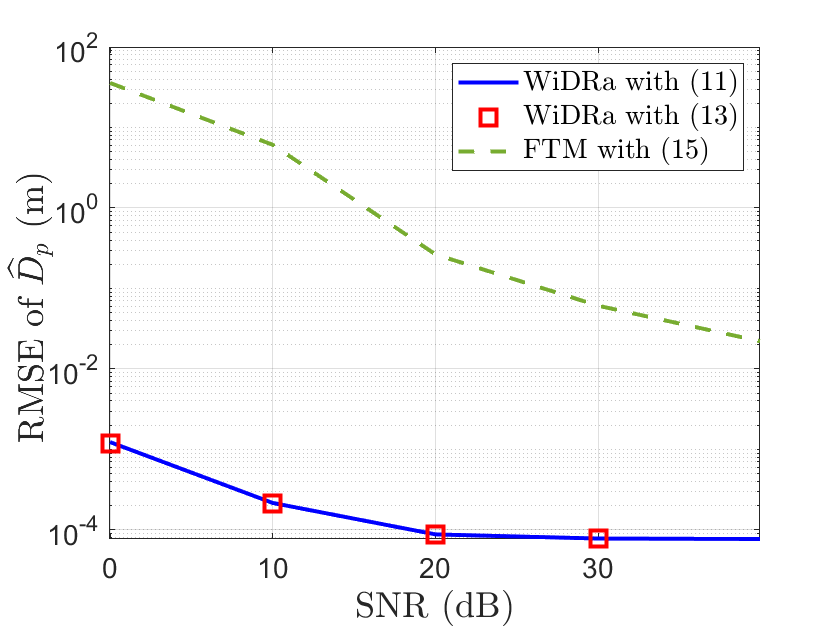} \label{Fig_RMSE_vs_SNR}}  \\
\subfloat[STA 2 speed: $S$]{\includegraphics[width= 0.4\textwidth]{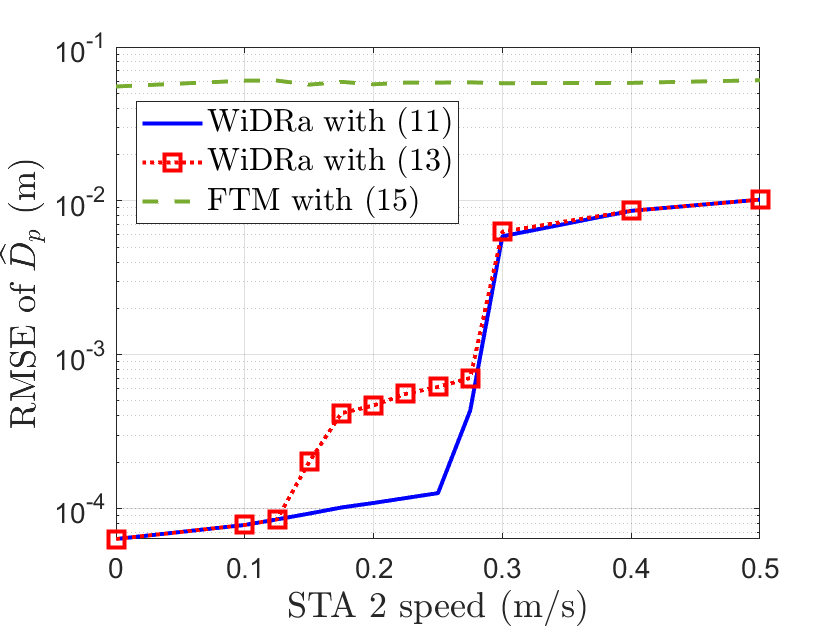} \label{Fig_RMSE_vs_Speed}} \\ 
\subfloat[K-factor: $\kappa$]{\includegraphics[width= 0.4\textwidth]{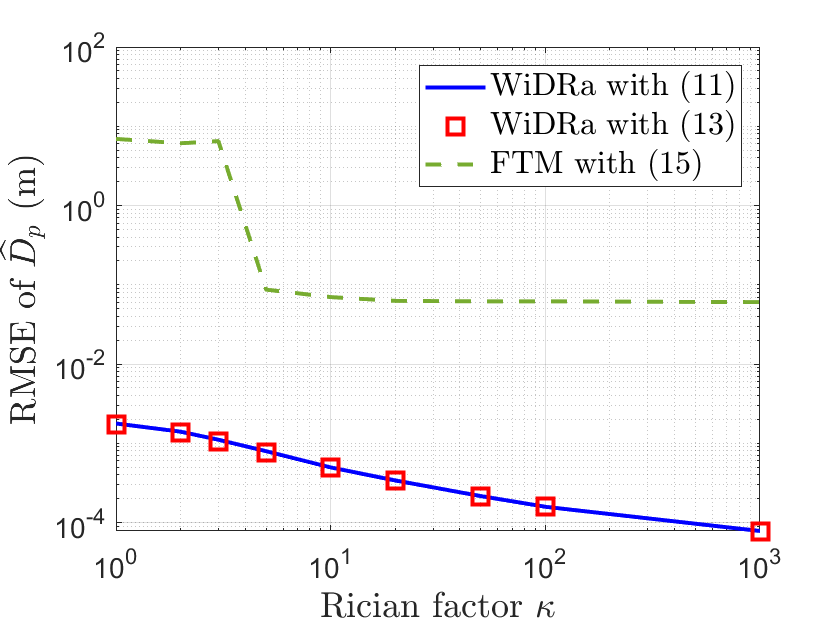} \label{Fig_RMSE_vs_Kfactor}} 
\caption{RMSE of sum-CP based differential range estimation under different simulation parameters. (a) depicts impact of SNR for $\kappa = 1000$, $S=0.1$ m/s. (b) depicts impact of $S$ for ${\rm SNR}=30$ dB and $\kappa=1000$. (c) depicts impact of Rician factor $\kappa$ for ${\rm SNR}=30$ dB and $S=0.1$ m/s.} 
\label{Fig_simulation_results}
\end{figure}
 
Next, we study the relative range estimates $\widehat{R}_{q,p}$ with WiDRa for samples separated by $t^{(1)}_p - t^{(1)}_q = T$ in Fig.~\ref{Fig_simulation_results_lags}. Fig.~\ref{Fig_RMSE_vs_lags} studies the relative range estimation RMSE versus $T$, and shows that the relative range estimation leads to only slow error accumulation over time of less than $1$ cm in $20$ s at $\kappa = 7$. We also observe that error accumulation is linear with $T$ (beyond $T=1$ s), suggesting that the only error source is cycle slips. In other words, the non-LoS paths cause \emph{local} errors in the differential range estimate that do not accumulate over time. This is justified further in Fig.~\ref{Fig_PDF_vs_lags}, where we plot the probability density function (PDF) of the relative range error $(\widehat{R}_{q,p} - {R}_{q,p})$ (modulo $\mathrm{c} \big/ 2Nf_{\rm c}$) for WiDRa with \eqref{eqn_dRange_est} as a function of the time gap $T = t^{(1)}_p - t^{(1)}_q$. Here the PDF converges just after $T=1$ s, suggesting that beyond that point the accumulated error variance due to non-LoS paths does not increase. Since cycle slips can be reduced by \eqref{eqn_dRange_est_2} and by reducing $t^{(1)}_p - t^{(1)}_{p-1}$, this is an encouraging result on the practicality of the proposed method for relative range estimation. However, further investigations may be needed for channels with diffuse scattering and non-stationary multi-path profiles.
\begin{figure}[!htb]
\centering
\subfloat[Time gap: $T$]{\includegraphics[width= 0.4\textwidth]{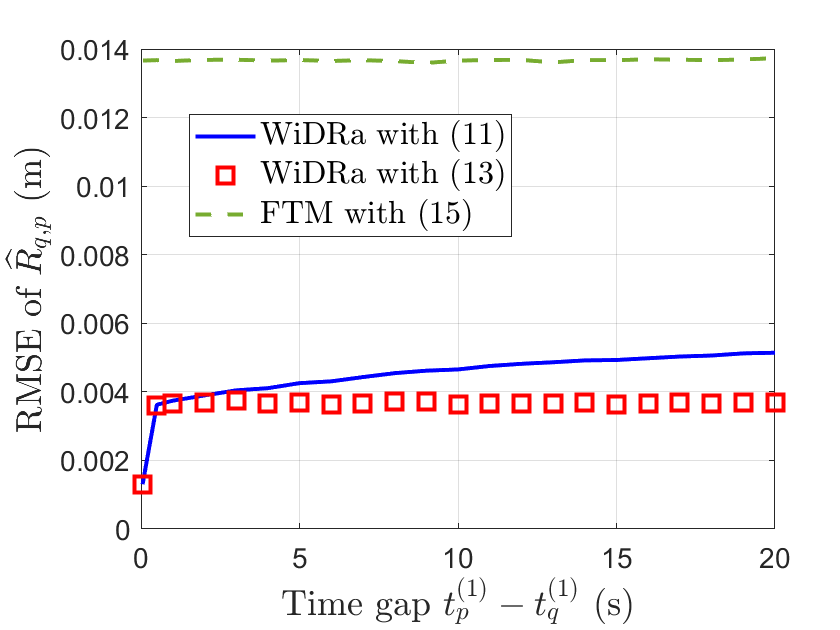} \label{Fig_RMSE_vs_lags}} \\
\subfloat[Probability Density Function (PDF)]{\includegraphics[width= 0.4\textwidth]{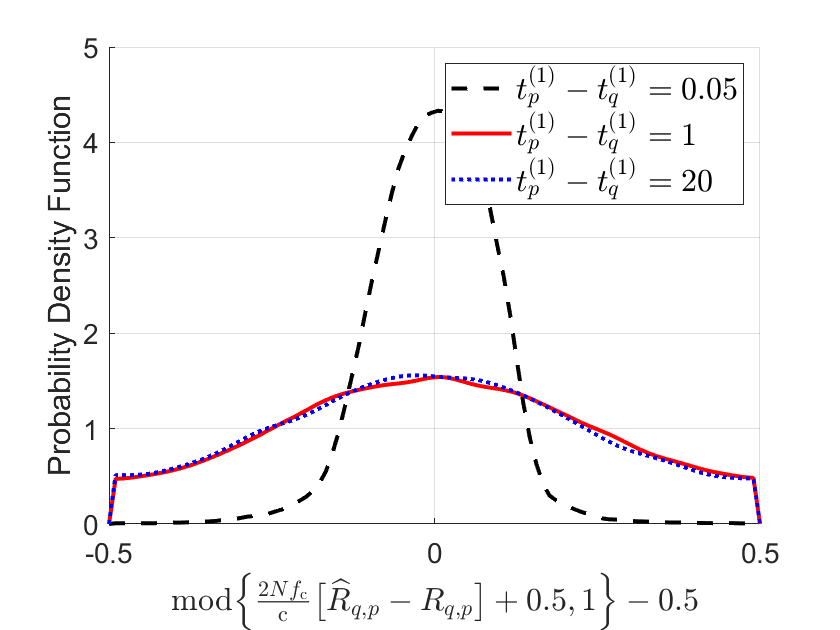} \label{Fig_PDF_vs_lags}} 
\caption{Statistics of sum-CP based relative range estimation error as a function of time gap $T$, for ${\rm SNR}=\infty$, $S=0.1$ m/s and $\kappa=7$. (a) depicts the RMSE of the error. (b) depicts the PDF of the error (modulo $\mathrm{c} \big/ 2Nf_{\rm c}$) for WiDRa with \eqref{eqn_dRange_est}. The PDFs with \eqref{eqn_dRange_est_2} are similar.} 
\label{Fig_simulation_results_lags}
\end{figure}

Finally, we summarize in Table \ref{Table_compute_time} the average computation time of the different steps of WiDRa and compare it to the computation time of FTM \eqref{eqn_FTM_Arange_est}. As is evident, the computation time is dominated by the root-MUSIC algorithm in Section \ref{subsec_symbol_timing_est}. Thus, use of low-complexity alternatives to root-MUSIC can reduce the computation time further. 
\begin{table}[h!]
\centering
\caption{Computation time (in ms) per frame exchange.} \label{Table_compute_time}
\begin{tabular}{| p{2.8cm} | p{1.5cm} | p{2.8cm} | } 
\hline 
Methods & Mean Time ($\mu$s) & Computation Complexity \\
\hline
Estimate $\widehat{\tau}_p^{(2)}, \widehat{\tau}_p^{(4)}$ (Section \ref{subsec_symbol_timing_est}) & $950$ & $\mathrm{O}(K^3)$ \\
\hline
CFO refinement (Section \ref{subsec_CFO_est}) & $7$ & $\mathrm{O}(F)$ \\
\hline
WiDRa with \eqref{eqn_dRange_est} & $957$ & $\mathrm{O}(2K^3) + \mathrm{O}(F)$ \\
\hline
WiDRa with \eqref{eqn_dRange_est_2} & $1000$ & $\mathrm{O}(2K^3) + \mathrm{O}(F) + \mathrm{O}(S_{\rm max} |\mathcal{Q}_p|)$ \\
\hline
FTM from \eqref{eqn_FTM_Arange_est} & $950$ & $\mathrm{O}(2K^3)$ \\
\hline
\end{tabular}
\end{table}

\subsection{Real-field data} \label{subsec_real_data}
\noindent To demonstrate proof-of-concept that WiDRa can indeed work, we also depict its performance on a few real-field data sets. 
For this, we set up a test bed involving two Intel AX210 devices connected to central processing units (CPUs). During the experiments, STA 1 is stationary and STA 2 is set-up on either (a) a programmable slider or (b) a cart that is moved with respect to STA 1 as shown in Fig.~\ref{Fig_setup_slider} and \ref{Fig_setup_cart}. The average distance between the STAs is set to $0.8$ m for slider, and $3$ m for cart setup. Although this ensures a strong LoS path, measurements are conducted in a typical office workspace with many metal doors, shelves, tables and chairs which act as scatterers and generate multi-path as depicted in the layouts of Fig.~\ref{Fig_real_field_images}. 
We perform 6 experiments each of which is run for $>60$ s, which are broadly of two types:
\begin{enumerate}
\item \emph{Slow movement -} In experiments 1-2, STA 2 is set up on the slider and is moved to and fro radially with respect to STA 1, at a speed of $< 50$ mm/s. 
\item \emph{Fast movement -} In experiments 3-6, STA 2 is set up on the cart that is radially moved away from STA 1 at speeds $> 50$ mm/s. 
\end{enumerate}
For performing the frame exchanges and acquiring the corresponding time stamps and CSI, we use the PicoScenes software \cite{Zhiping2022} in the Initiator-Responder mode. 
We use only one of the $2$ antennas on each STA for the transmission/reception of frames. The time stamps and CSI $\{ t_p^{(1)}, t_p^{(4)}, \bar{h}^{(2)}_{p,k} | p \in \mathcal{P}, k \in \mathcal{K} \}$ at STA 1 and $\{ t_p^{(2)},t_p^{(3)}, \bar{h}^{(2)}_p | p \in \mathcal{P}, k \in \mathcal{K} \}$ at STA 2 are logged on the CPUs and the proposed algorithms are tested in post-processing. As shown in Fig.~\ref{Fig_CFO_impact_illustrate}, the Intel AX210 chip suffers from an arbitrary phase rotation of $N=2$. 
It should be noted that the PicoScenes has several limitations:
\begin{itemize}
\item It doesn't allow good control of the frame transmission times and so pattern in Fig.~\ref{Fig_nested_array_jitter} can't be replicated. As a work around, we perform frame exchanges at the fastest allowed rate of $t_p^{(1)} - t_{p-1}^{(1)} \approx 1$ ms and then sub-sample the data to the desired rate of $25$ ms spacing. To allow refined CFO estimation (Section \ref{subsec_CFO_est}), we do not under sample the first $20$ samples of each $W=0.5$ s window. This under-sampling process is depicted pictorially in Fig.~\ref{Fig_subsampling}. It should however be noted that the average CFO between the two devices is around $4-8$ KHz, which is still heavily under-sampled at this $1$ ms spacing making CFO refinement challenging. 
\item The shortest round-trip time it allows is $t_p^{(4)}-t_p^{(1)} \approx 0.5$ms, which leads to a significant amount of residual phase noise that impacts the estimation and can cause cycle slips (see Fig.~\ref{Fig_PhN_impact_illustrate}).
\item It occasionally experiences $>100$ ms intervals of no transmissions, which make tracking the relative range harder and can cause cycle slips.
\item It only provides transmit time stamps with $25$ ns precision. We therefore skip the comparison to FTM \eqref{eqn_FTM_Arange_est} which requires a much finer precision.
\end{itemize}
Thus the results presented below are constrained, and can be further improved with a more advanced tool. 
\begin{figure}[!htb]
\centering
\subfloat[Slider setup]{\includegraphics[width= 0.45\textwidth]{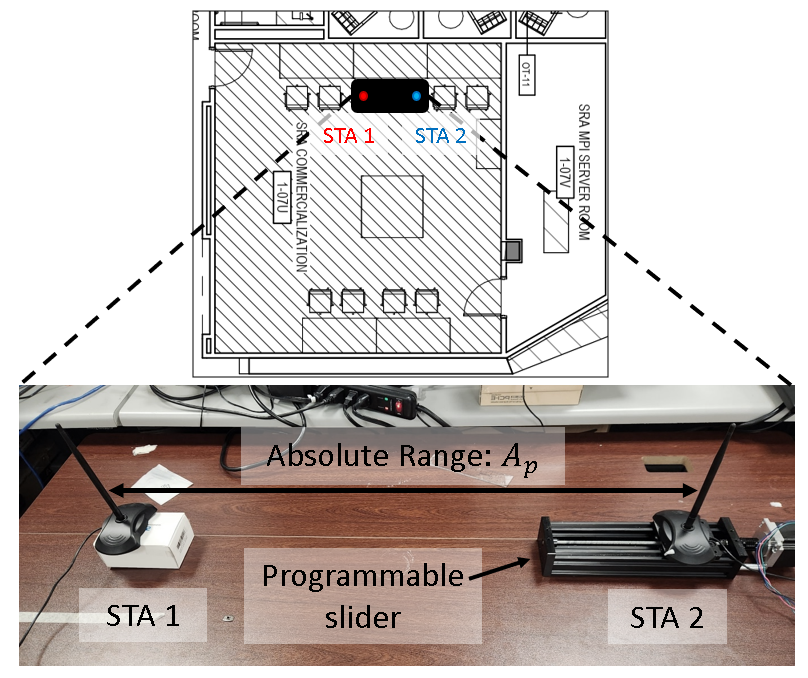} \label{Fig_setup_slider}} \\
\subfloat[Cart setup]{\includegraphics[width= 0.45\textwidth]{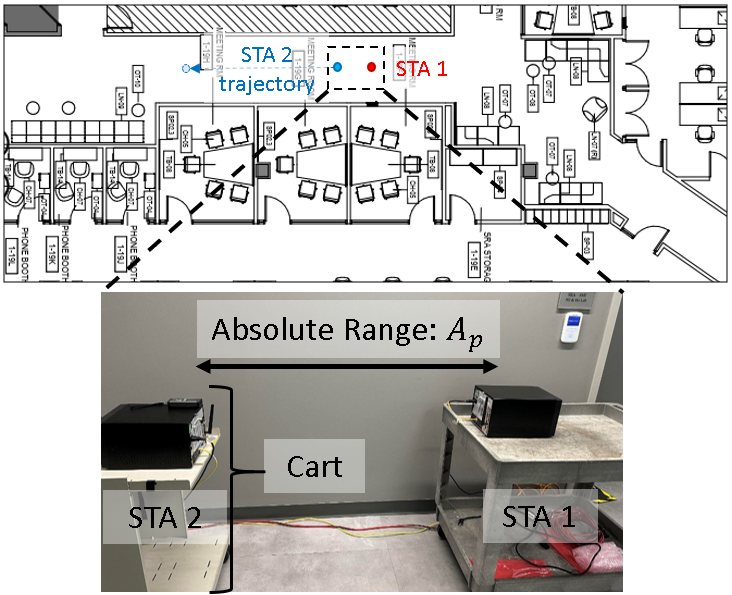} \label{Fig_setup_cart}} \\
\subfloat[Sub-sampling data]{\includegraphics[width= 0.48\textwidth]{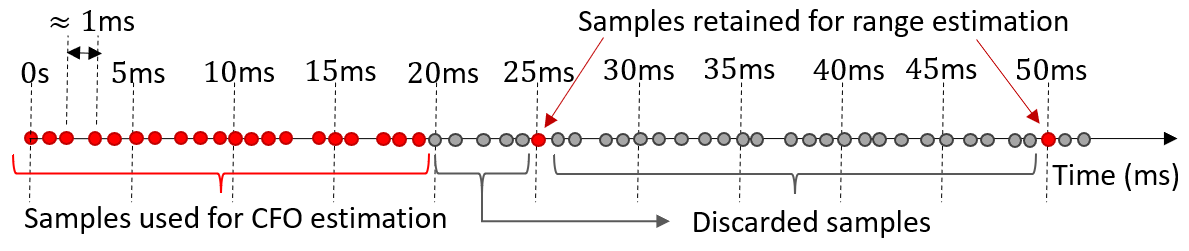} \label{Fig_subsampling}} 
\caption{Measurement setup of STA 1 and STA 2 for the real-field data with (a) slider-based movement and (b) cart-based movement. The sub-sampling of the frame exchanges to approximate the nested array pattern is depicted in (c).} 
\label{Fig_real_field_images}
\end{figure}

For each experiment, the RMSE of (i) differential range $\widehat{D}_{p}$ and (ii) the relative range $\widehat{R}_{q,p}$ for a specific time separation $t^{(1)}_p - t^{(1)}_q = T$, obtained from \eqref{eqn_dRange_est_2} are depicted in Table~\ref{Table2}. 
\begin{table}[h!]
\centering
\caption{RMSE for differential and relative range using \eqref{eqn_dRange_est_2}.} \label{Table2}
\begin{tabular}{| p{0.5cm} | p{0.6cm} | p{0.6cm} | p{0.65cm} | p{0.55cm} | p{0.6cm} | p{0.6cm} | p{0.6cm} |} 
\hline 
\multirow{3}{0.5cm}{Exp. No.} & \multirow{3}{0.6cm}{$f_{\rm c}$ MHz} & \multirow{3}{0.6cm}{BW MHz} & \multirow{3}{0.65cm}{Speed mm/s} & \multicolumn{4}{c |}{RMSE [mm]} \\
\cline{5-8}
& & & & \multirow{2}{0.55cm}{$\widehat{D}_p$} & \multicolumn{3}{c |}{$\{ \widehat{R}_{q,p} | t^{(1)}_p-t^{(1)}_q = T\}$}\\
\cline{6-8}
& & & & & $T=50$ms & $T=1$s & $T=5$s \\ 
\hline
1 & $5260$ & $20$ & $5$ & $0.2$ & $0.2$ & $0.7$ & $0.9$  \\ 
\hline
2 & $5260$ & $20$ & $30$ & $0.3$ & $0.4$ & $2.4$ & $4.2$ \\ 
\hline
3 & $5985$ & $80$ & $55$ & $0.6$ & $1.3$ & $16.5$ & $48.4$ \\ 
\hline
4 & $5985$ & $80$ & $100$ & $0.7$ & $1.6$ & $20.3$ & $53.1$ \\ 
\hline
5 & $6015$ & $20$ & $140$ & $1$ & $1.9$ & $22$ & $43.8$ \\ 
\hline
6 & $6015$ & $20$ & $350$ & $2.2$ & $5$ & $80.2$ & $269.4$ \\ 
\hline
\end{tabular}
\end{table}
From Remark \ref{remark_max_speed}, note that the maximum resolvable speed using \eqref{eqn_dRange_est_2} is $\approx 250$ mm/s for $25$ ms sampling. We observe that closer the speed of STA 2 is to this limit, the higher is the chance for cycle slips and, correspondingly, the RMSE is higher. At very low speeds, we observe that even after an interval of $5$ s, the cumulative error in relative range estimate is below $1$ cm. This is also evident from Fig.~\ref{Fig_relative_range_plot}, where we plot the relative range $\widehat{R}_{1,p}$ trajectories predicted by \eqref{eqn_dRange_est} and \eqref{eqn_dRange_est_2} for two sample experiments. Thus, the results conclusively demonstrate that differential and relative range estimation with mm-level accuracy is indeed feasible in Wi-Fi systems by tracking the carrier phase in LoS channels. 
\begin{figure}[!htb]
\centering
\subfloat[Exp. No. 2 (slider setup)]{\includegraphics[width= 0.4\textwidth]{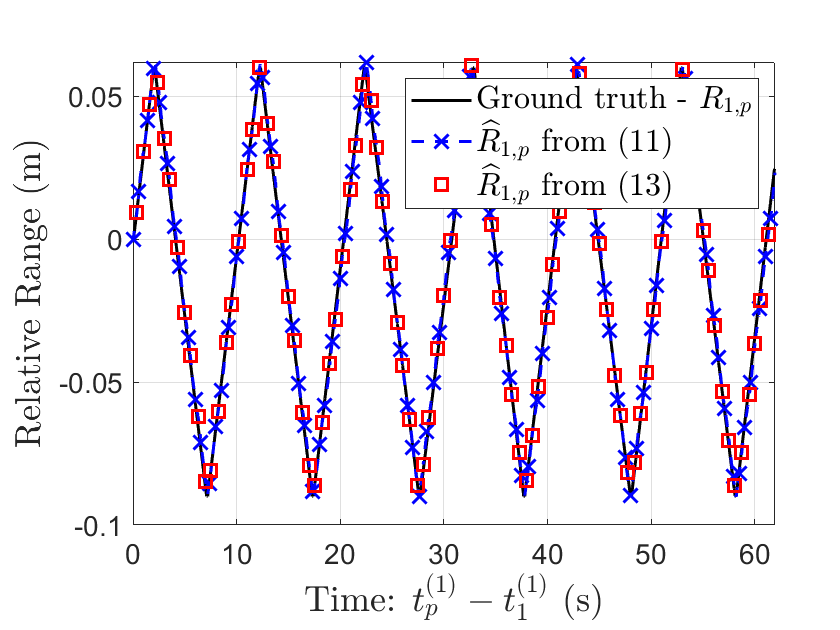} \label{Fig_rel_range_file6}} \\
\subfloat[Exp. No. 5 (cart setup)]{\includegraphics[width= 0.4\textwidth]{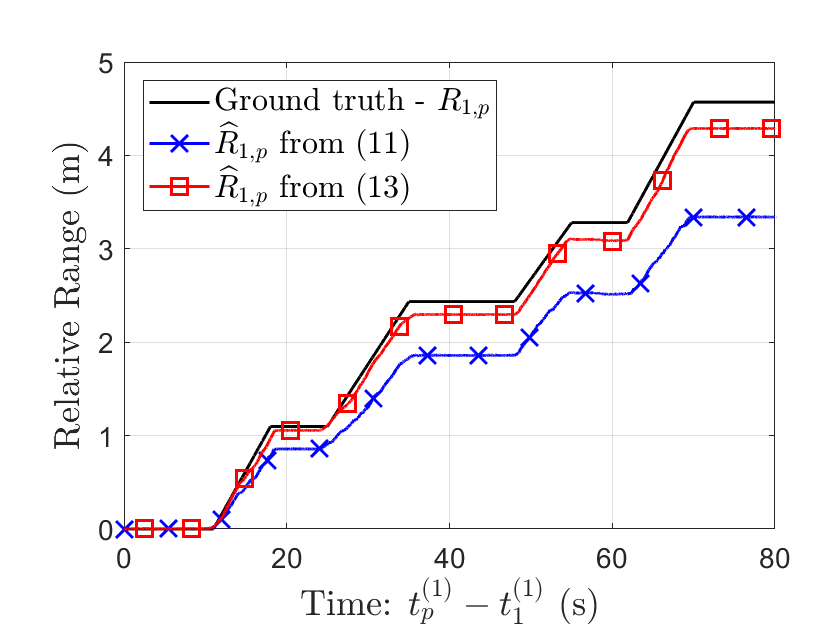} \label{Fig_rel_range_file2}}
\caption{A plot of the estimated relative range from \eqref{eqn_rel_range} compared to the ground truth.} 
\label{Fig_relative_range_plot}
\end{figure}

\section{Limitations and future directions} \label{sec_future_dir}
\noindent Although this paper introduces the potential of exploiting CP for differential ranging with Wi-Fi and provides some proof-of-concept results, it has to be emphasized that there is a vast range of topics left to investigate. Here we highlight a few of them to help stimulate research in these directions. 

\subsubsection{Dealing with multi-path channels} 
The analysis in this paper assumed an ideal LoS channel. Although the simulation results show good performance with multi-path, further work and evaluations are needed for channels with small Rician factor ($\kappa$). In particular, the use of the method in non-LoS channels is challenging, where the direct path is significantly attenuated leading to small values of $\kappa$. A simple solution is use WiDRa opportunistically when $\kappa$ is large, and to rely more on the RTT-based range estimate when $\kappa$ is small. When multiple antennas and wider transmission bandwidths are available at the STAs, one may also investigate the use of the spatial and frequency dimensions in CSI to better resolve the LoS path from the other paths when estimating $\widehat{\tau}_p^{(\bullet)}$ in Section \ref{subsec_symbol_timing_est}. This can reduce the impact of the other paths on \eqref{eqn_psi_correct_music} and improve performance even when $\kappa$ is small. 

\subsubsection{Dealing with cycle slips}
One concern with the use of CP for relative range estimation is the presence of cycle slips which can lead to accumulative error over time. Further work may be required to reduce the probability of cycle slips using, for example, FTM-based absolute range information, sub-Nyquist sampling methods \cite{Venkataramani2001}, or CSI from multiple antennas or links (in Wi-Fi 7).

\subsubsection{Frame exchange overhead and coexistence}
To enable unambiguous tracking for a maximum STA 2 speed of $1$ m/s, an average needed inter-frame exchange time at $5$ GHz is $15$ ms (for $N=1$) which leads to a very large communication overhead. The channel contention in Wi-Fi which can cause a significant delay in the channel access between adjacent frames is also another limiting factor which can reduce the maximum tracking speed. Mechanisms to reduce the required communication overhead and deal with the channel access delay is a good direction for future investigation.  

\subsubsection{Passive CP ranging}
In IEEE 802.11az standard \cite{IEEE_11az}, passive ranging protocol was introduced where STA 2 doesn't actively participate in the frame exchanges but, instead, listens to the frame exchanges between anchor STAs passively to locate itself. This method scales well with increasing number of devices. Another good direction of investigation is the expansion of the proposed method to passive ranging. 

\subsubsection{Improving absolute range estimation}
Although the focus of this paper was on obtaining the differential range, differential range estimates can also be used to improve absolute range estimates obtained from, say, FTM protocol. This is because knowledge of the differential range enables intelligent smoothing of an absolute range estimate from its neighboring absolute range estimates. This can be investigated further, leveraging the rich literature on sensor fusion \cite{Wu2019} to combine the CP-based differential range with the RTT-based absolute range. 

\subsubsection{Direction-finding and triangulation} 
Note from \eqref{eqn_Phi_diff} that, with accurate CFO estimates, difference in phase-sum $\Psi_{p}-\Psi_{p-a}$ can be interpreted as the relative change in phase of a carrier of frequency $2 f_{\rm c}$ as received at two antennas that are located at the true locations of STA 2 at times $t_p^{(1)}$ and $t_{p-a}^{(1)}$, respectively. In other words, the proposed sum-CP method enables coherent phase measurements of STA 2 as it moves, thus generating a \emph{virtual antenna array}. Therefore sum-CP measurements can potentially enable direction finding and triangulation-based localization \cite{Zamora2002, Xiong2013} with single antenna STAs, which is worth investigating. 

\section{Conclusion} \label{sec_conclusions}
\noindent This paper is the first to explore the idea of performing differential ranging between two devices using the carrier phase of un-synchronized systems like Wi-Fi. As shown from the analysis, the ranging accuracy is only limited by the carrier frequency and not the bandwidth. The appropriate metric that can track the range is the sum-CP, viz., sum of the phase of CSI measured at the carrier frequency at the two devices, after compensating for the symbol start time errors and the CFO. As shown from the results, a piece-wise fit to the CFO for every $0.5$ s provides sufficient resolution for CFO removal from the sum-CP. The analysis also suggests the use of a nested array transmission structure for the CP Request and Response frames to enable simultaneous estimation of CFO and the differential range. We also conclude that for a given frame exchange rate, the maximum device speed that can be tracked accurately is bounded, beyond which cycle slips are observed. Simulation results suggest that with $25$ ms frame exchange interval, sum-CP based differential range can outperform RTT-based differential range by $100 \times$, and can even perform well in multi-path channels with the Rician factor above $\kappa=7$. Real-world experiments validate the theoretical analysis and provide conclusive evidence that carrier phase can indeed provide millimeter-level differential ranging and relative ranging capability. The analysis also highlights the limitations of the method such as cycle slips, accommodating faster device speeds and high multi-path channels etc., and some pointers to deal with these issues are suggested.

\begin{appendices}
\section{} \label{appdix_csi_acq_frame}
\begin{proof}[Proof of Lemma \ref{Lemma_CSI_acq}]
The transmit signal from STA 1 for the $m$-th symbol of the CP Request frame $p$ can be written as 
\begin{eqnarray}
s(t) = \sum_{k \in \mathcal{K}} x_{k} e^{{\rm j} 2 \pi \frac{k}{T_{\rm s}} (t - t_p^{(1)} - m (T_{\rm s}+T_{\rm cy}))} e^{{\rm j} 2 \pi f_{\rm c} t}, \label{eqn_tx_sig_s}
\end{eqnarray}
for $-T_{\rm cy} \leq t - t_p^{(1)} - m (T_{\rm s}+T_{\rm cy}) < T_{\rm s}$, where $x_{k}$ is the pilot data on the $k$-th subcarrier of symbol $m$. Let the error at STA 2 in detection of the symbol start time for the $p$-th CP Request frame be $\tau_p^{(2)}$. Then the perceived time of reception of the $m$-th OFDM symbol at STA 2 (measured in STA 1's time reference) is 
\begin{eqnarray}
\bar{T} = t_p^{(1)} + \rho_p + \tau_p^{(2)} + m \frac{(T_{\rm s}+T_{\rm cy})}{1+ \eta_p} . 
\end{eqnarray}
Let us assume that STA 2 estimates the CSI from the $m$-th OFDM symbol of the frame.\footnote{This symbol can either by the legacy long training field or the class-specific long training field of the frame.}. The OFDM demodulation output on sub-carrier $k$ for this symbol can then be expressed as 
\begin{eqnarray}
Y_k &=& \int_{t=\bar{T}}^{\bar{T}+\frac{T_{\rm s}}{1+ \eta_p}} \frac{1}{T_{\rm s}} \Big[ s (t - \rho_p) e^{-{\rm j} 2 \pi f_c(1+\eta_p)t} e^{{\rm j} 2 \pi f_c \eta_p t_p^{(1)}} \nonumber \\
&& \times e^{-{\rm j} \phi_p} e^{{\rm j} 2 \pi \widehat{f}_{\rm CFO}(t- t_p^{(1)} - \rho_p - \tau_p^{(2)})} \nonumber \\
&& \times e^{- {\rm j} 2 \pi \frac{k}{T_{\rm s}}(1+\eta_p) (t-\bar{T})} \Big] {\rm d}t , \label{eqn_demod_Y_k}
\end{eqnarray}
where $\widehat{f}_{\rm CFO}$ is the CFO estimated at STA 2 from the received frame \cite{Sourour2004}. Assuming that STA 2 can accurately estimate the CFO value, i.e., $\widehat{f}_{{\rm CFO}} = \eta_p f_{\rm c}$, \eqref{eqn_demod_Y_k} can be further simplified as
\begin{eqnarray}
Y_k & = & \int_{t=\bar{T}}^{\bar{T}+\frac{T_{\rm s}}{1+ \eta_p}} \frac{1}{T_{\rm s}} \bigg[ s (t - \rho_p) e^{-{\rm j} 2 \pi f_c t} e^{ - {\rm j} 2 \pi \eta_p f_{\rm c} (\rho_p + \tau_p^{(2)})} \nonumber \\
&& \times e^{-{\rm j} \phi_p} e^{- {\rm j} 2 \pi \frac{k}{T_{\rm s}}(1+\eta_p) (t-\bar{T})} \bigg] {\rm d}t  \nonumber \\
&=& \int_{t=\bar{T}}^{\bar{T}+\frac{T_{\rm s}}{1+ \eta_p}} \sum_{\bar{k} \in \mathcal{K}}  \frac{x_{\bar{k}}}{T_{\rm s}} \bigg[ e^{{\rm j} 2 \pi \frac{\bar{k}}{T_{\rm s}} (t - t_p^{(1)} - \rho_p - m (T_{\rm s}+T_{\rm cy}))} e^{-{\rm j} \phi_p} \nonumber \\
&& \times e^{- {\rm j} 2 \pi f_{\rm c} [(1+\eta_p)\rho_p + \eta_p \tau_p^{(2)}]} e^{- {\rm j} 2 \pi \frac{k}{T_{\rm s}}(1+\eta_p) (t-\bar{T})} \bigg] {\rm d}t , \label{eqn_demod_Y_k_2}
\end{eqnarray}
where we use \eqref{eqn_tx_sig_s}. Using the fact that $|\eta_p m K| \ll 1$, and $|\eta_p f_{\rm c} \tau_p^{(2)}| \ll 1$ under typical scenarios, \eqref{eqn_demod_Y_k_2} can be further simplified to
\begin{eqnarray}
Y_k & \approx & \int_{t=\bar{T}}^{\bar{T}+T_{\rm s}} \sum_{\bar{k} \in \mathcal{K}}  \frac{x_{\bar{k}}}{T_{\rm s}} \bigg[ e^{{\rm j} 2 \pi \frac{\bar{k}}{T_{\rm s}} (t + \tau_p^{(2)} - \bar{T})} e^{- {\rm j} 2 \pi \frac{k}{T_{\rm s}} (t-\bar{T})} \nonumber \\
&& \times e^{- {\rm j} 2 \pi f_{\rm c} (1+\eta_p)\rho_p} e^{-{\rm j} \phi_p} \bigg] {\rm d}t \nonumber \\
&=& x_k \bigg[ e^{{\rm j} 2 \pi \frac{k}{T_{\rm s}} \tau_p^{(2)}} e^{- {\rm j} 2 \pi f_{\rm c} (1+\eta_p)\rho_p } e^{-{\rm j} \phi_p} \bigg] .
\end{eqnarray}
Dividing the right hand side by $x_k$, the result follows.
\end{proof}

\section{} \label{appdix_csi_resp_frame}
\begin{proof}[Proof of Lemma \ref{Lemma_CSI_resp}]
Let the error at STA 1 in detection of the symbol start time for the $p$-th CP Response frame be $\tau_p^{(4)}$. Then the true time of transmission of the $p$-th CP Response frame (measured in STA 1's time reference) is $t_p^{(4)} - \rho_p - \Delta\rho - \tau_p^{(4)}$. Then the transmit signal from STA 2 for the $m$-th symbol of the CP Response frame $p$ can be written as
\begin{eqnarray} \label{eqn_tx_sig_s_alt}
s(t) &=& \sum_{k \in \mathcal{K}} x_{k} e^{{\rm j} 2 \pi \frac{(1+\eta_p)k}{T_{\rm s}} \left(t - t_p^{(4)} + \rho_p + \Delta\rho + \tau_p^{(4)} - m \frac{(T_{\rm s}+T_{\rm cy})}{1+\eta_p} \right)} \nonumber \\
&& \times e^{{\rm j} 2 \pi f_c(1+\eta_p)t} e^{-{\rm j} 2 \pi f_c \eta_p t_p^{(1)}} e^{{\rm j} \phi_p},
\end{eqnarray}
for $$-T_{\rm cy} \leq t - t_p^{(4)} + \rho_p + \Delta\rho + \tau_p^{(4)} - m \frac{(T_{\rm s}+T_{\rm cy})}{1+\eta_p} < T_{\rm s}, $$ where $x_{k}$ is the pilot data on the $k$-th subcarrier of symbol $m$. Let us assume that STA 1 estimates the CSI from the $m$-th OFDM symbol of the frame.\footnote{This symbol can either by the legacy long training field or the class-specific long training field of the frame.}. Note that the perceived time of reception of the $m$-th OFDM symbol at STA 1 (measured in STA 1's time reference) is $\bar{T} = t_p^{(4)} + m(T_{\rm s} + T_{\rm cy})$. The OFDM demodulation output on sub-carrier $k$ for this symbol can then be expressed as 
\begin{eqnarray}
Y_k &=& \int_{t=\bar{T}}^{\bar{T}+T_{\rm s}} \frac{1}{T_{\rm s}} \bigg[ s (t - \rho_p - \Delta\rho) e^{-{\rm j} 2 \pi f_c t} \nonumber \\
&& \times e^{{\rm j} 2 \pi \widehat{f}_{\rm CFO}(t- t_p^{(4)})} e^{- {\rm j} 2 \pi \frac{k}{T_{\rm s}} (t-\bar{T})} \bigg] {\rm d}t , \label{eqn_demod_Y_k_alt}
\end{eqnarray}
where $\widehat{f}_{\rm CFO}$ is the CFO estimated at STA 1 from the received frame \cite{Sourour2004}. Assuming that STA 1 can accurately estimate the CFO value, i.e., $\widehat{f}_{{\rm CFO}} = -\eta_p f_{\rm c}$, \eqref{eqn_demod_Y_k_alt} can be further simplified as 
\begin{eqnarray}
Y_k &=& \int_{t=\bar{T}}^{\bar{T}+T_{\rm s}} \sum_{\bar{k} \in \mathcal{K}} \frac{x_k}{T_{\rm s}} \bigg[ e^{{\rm j} 2 \pi \frac{(1+\eta_p)\bar{k}}{T_{\rm s}} \left(t - t_p^{(4)} + \tau_p^{(4)} - m \frac{(T_{\rm s}+T_{\rm cy})}{1+\eta_p} \right)} \nonumber \\
&& \times e^{-{\rm j} 2 \pi f_{\rm c}(1+\eta_p) (\rho_p+\Delta\rho)} e^{- {\rm j} 2 \pi \eta_p f_{\rm c} t_p^{(1)}} \nonumber \\
&& \times e^{{\rm j} \phi_p} e^{{\rm j} 2 \pi \eta_p f_{\rm c} t_p^{(4)}} e^{- {\rm j} 2 \pi \frac{k}{T_{\rm s}} (t-\bar{T})} \bigg] {\rm d}t , \label{eqn_demod_Y_k_2_alt}
\end{eqnarray}
where we use \eqref{eqn_tx_sig_s_alt}. Using the fact that $|\eta_p m K| \ll 1$, \eqref{eqn_demod_Y_k_2_alt} can be further simplified to 
\begin{eqnarray}
Y_k & \approx & \int_{t=\bar{T}}^{\bar{T}+T_{\rm s}} \sum_{\bar{k} \in \mathcal{K}} \frac{x_k}{T_{\rm s}} \bigg[ e^{{\rm j} 2 \pi \frac{\bar{k}}{T_{\rm s}} \left(t - \bar{T} + \tau_p^{(4)} \right)} e^{{\rm j} 2 \pi \eta_p f_{\rm c} [t_p^{(4)}-t_p^{(1)}]} \nonumber \\
&& \times e^{-{\rm j} 2 \pi f_{\rm c}(1+\eta_p) (\rho_p+\Delta\rho)} e^{{\rm j} \phi_p} e^{- {\rm j} 2 \pi \frac{k}{T_{\rm s}} (t-\bar{T})} \bigg] {\rm d}t , \nonumber \\
&=& x_k \bigg[ e^{{\rm j} 2 \pi \frac{k}{T_{\rm s}} \tau_p^{(4)}} e^{- {\rm j} 2 \pi f_{\rm c} (1+\eta_p)(\rho_p+\Delta\rho) } \nonumber \\
&& \times e^{{\rm j} \phi_p} e^{{\rm j} 2 \pi \eta_p f_{\rm c} [t_p^{(4)}-t_p^{(1)}]} \bigg] .
\end{eqnarray}
Dividing the right hand side by $x_k$, the result follows.
\end{proof}

\end{appendices}




%

\bibliographystyle{IEEEtran}
\bibliography{IEEEabrv, references_rev2_abrev}

\begin{IEEEbiography}[{\includegraphics[width=1.1in,height=1.25in,clip,keepaspectratio]{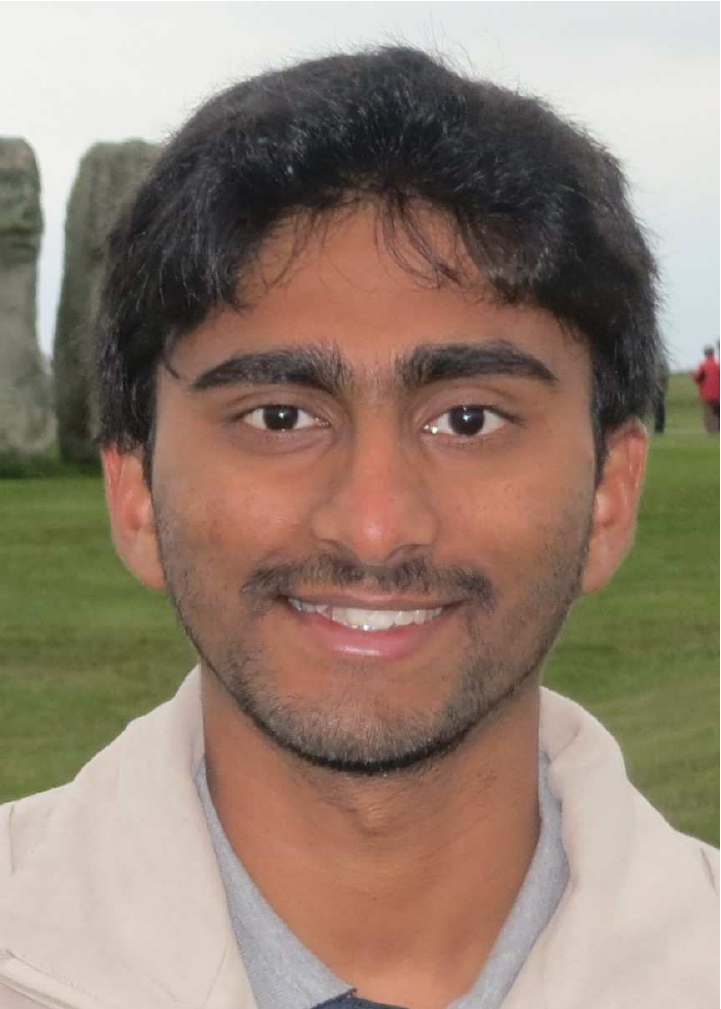}}]{Vishnu V. Ratnam} (S'10--M'19--SM'22) received the B.Tech. degree (Hons.) in electronics and electrical communication engineering from IIT Kharagpur, Kharagpur, India in 2012, where he graduated as the Salutatorian for the class of 2012. He received the Ph.D. degree in electrical engineering from University of Southern California, Los Angeles, CA, USA in 2018. He currently works as a Staff Research Engineer II in the Standards and Mobility Innovation Lab at Samsung Research America, Plano, Texas, USA. His research interests are in Wi-Fi standards, wireless sensing, AI for wireless, mm-Wave and Terahertz communication. 
Dr. Ratnam was the recipient of the Best Student Paper Award with the IEEE International Conference on Ubiquitous Wireless Broadband (ICUWB) in 2016, the Bigyan Sinha memorial award in 2012 and is a member of the Phi-Kappa-Phi honor society.
\end{IEEEbiography}

\begin{IEEEbiography}[{\includegraphics[width=1.1in,height=1.25in,clip,keepaspectratio]{ 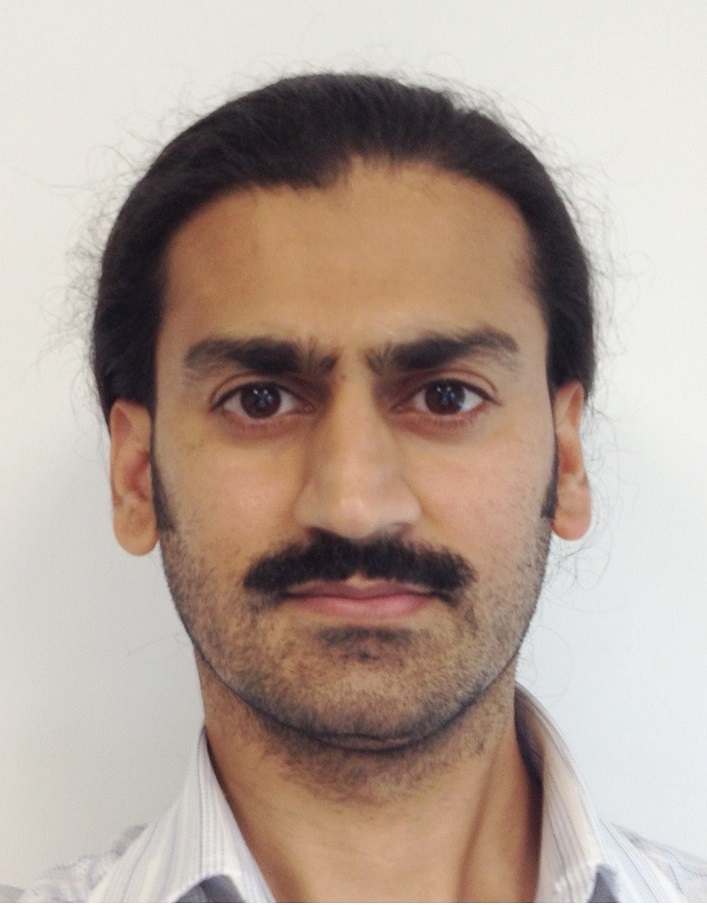}}]{Bilal Sadiq} received his Ph.D. degree in electrical and computer engineering from The University of Texas at Austin in 2010. He is currently with Standards and Mobility Innovation Laboratory, Samsung Research America. His research includes applications of network asymptotics, queueing theory, signal processing, optimization theory, and ML in wireless communication systems. Dr. Sadiq was the recipient of Best Paper Award at ITC-22 in September 2010.
\end{IEEEbiography}

\begin{IEEEbiography}[{\includegraphics[width=1in,height=1.25in,clip,keepaspectratio]{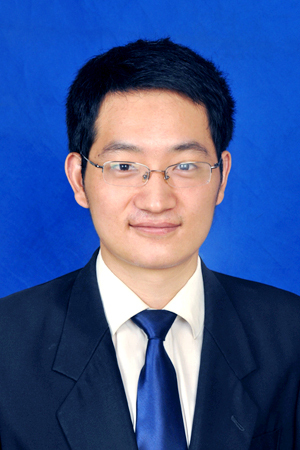}}]{Hao Chen} received his B.S. and M.S. degrees in Information Engineering from Xi'an Jiaotong University, Shaanxi, in 2010 and 2013. He received the Ph.D. degree in Electrical Engineering from University of Kansas, Lawrence, KS, in 2017. He currently works as a Senior Staff Engineer with the Standards and Mobility Innovation Laboratory, Samsung Research America, where he is working on algorithm design and prototyping of AI for wireless communication, wireless sensing, and localization. His research interests include network optimization, machine learning, and 5G cellular systems.
\end{IEEEbiography}

\begin{IEEEbiography}[{\includegraphics[width=1in,height=1.25in,clip,keepaspectratio]{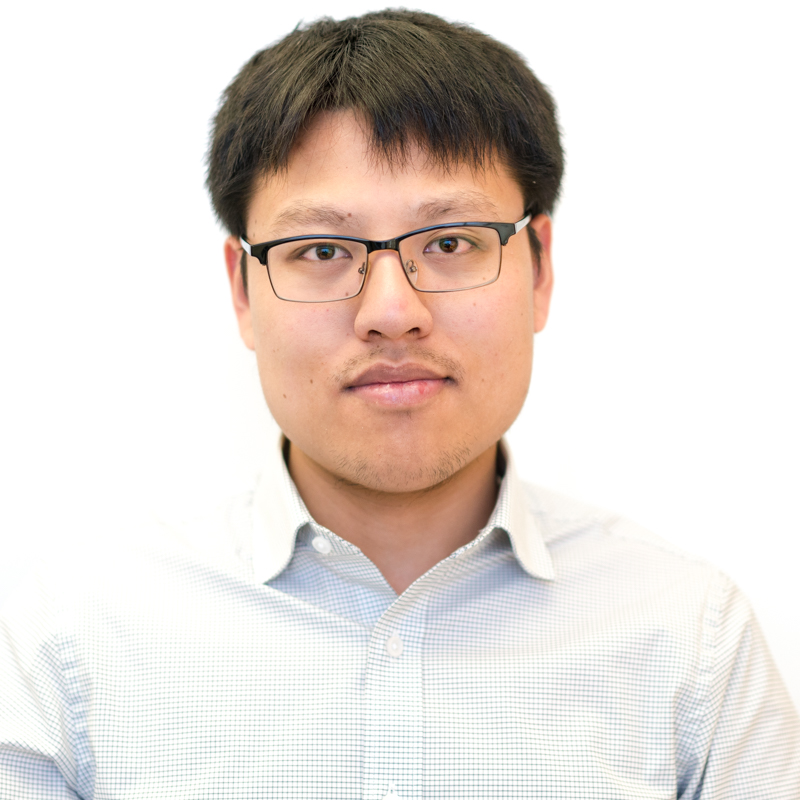}}]{Wei Sun} received his B.S. degrees in Computer Science from Shanghai Jiao Tong University, Shanghai, in 2017. He received the Ph.D. degree in from University of Texas at Austin, Austin, TX, in 2022. He currently works as a Senior Research Engineer with the Standards and Mobility Innovation Laboratory, Samsung Research America. He is working on algorithm design and prototyping of AI for wireless, wireless sensing and localization, audio recognition and detection. His research interests include foundation models and multi-modality modelling. 
\end{IEEEbiography}

\begin{IEEEbiography}[{\includegraphics[width=1in,height=1.25in,clip,keepaspectratio]{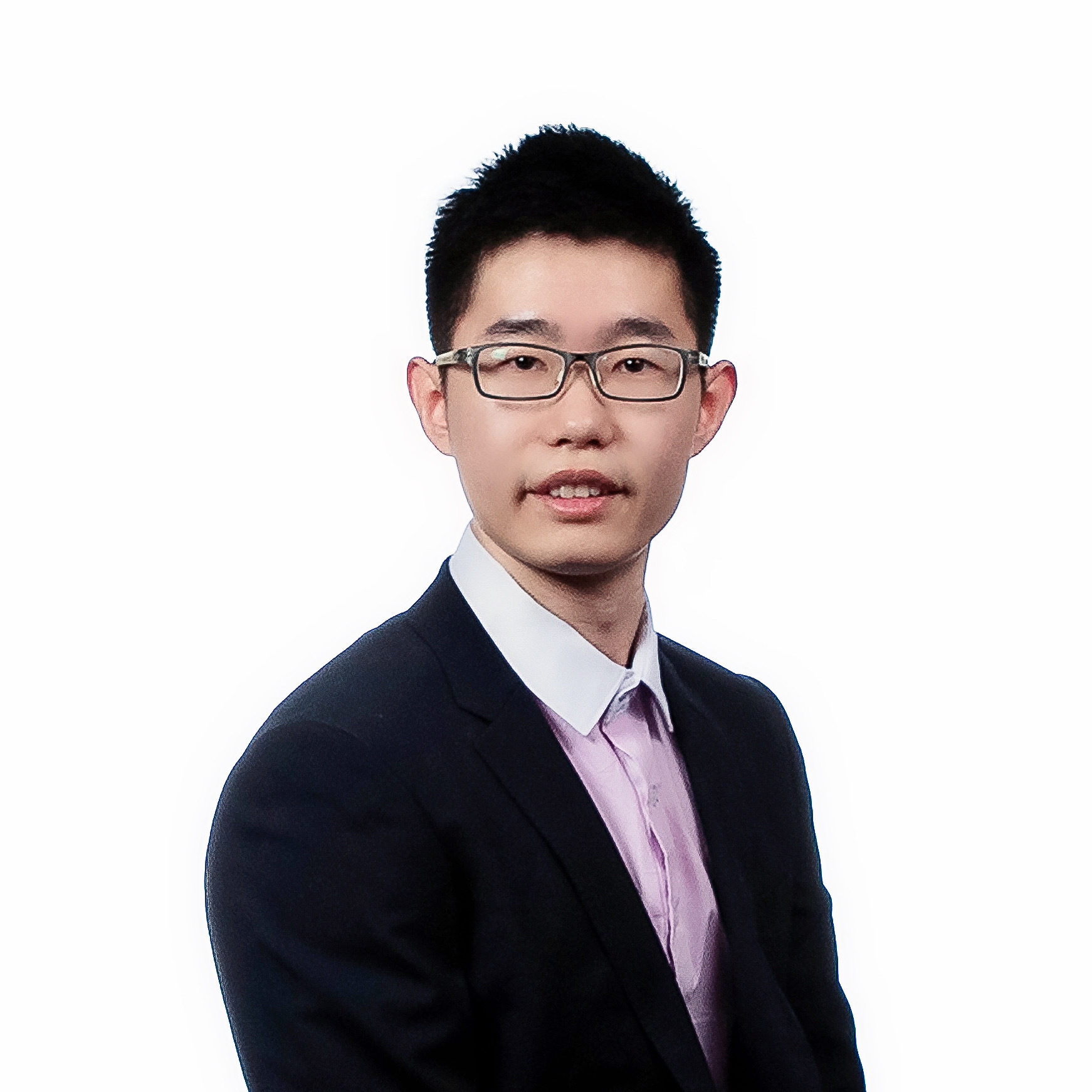}}]{Shunyao Wu} received his B.S. degrees in Electrical Engineering from Huazhong University of Science and Technology, Wuhan, in 2015. He received the M.S. and Ph.D. degree in Electrical Engineering from Arizona State University, Tempe, AZ, in 2017 and 2022. He currently works as a Senior Engineer with the Standards and Mobility Innovation Laboratory, Samsung Research America, where he is working on algorithm design and prototyping of modem system. His research interests include machine learning, 5G cellular systems algorithm design and hardware implementation.
\end{IEEEbiography}

\begin{IEEEbiography}[{\includegraphics[width=1in,height=1.25in,clip,keepaspectratio]{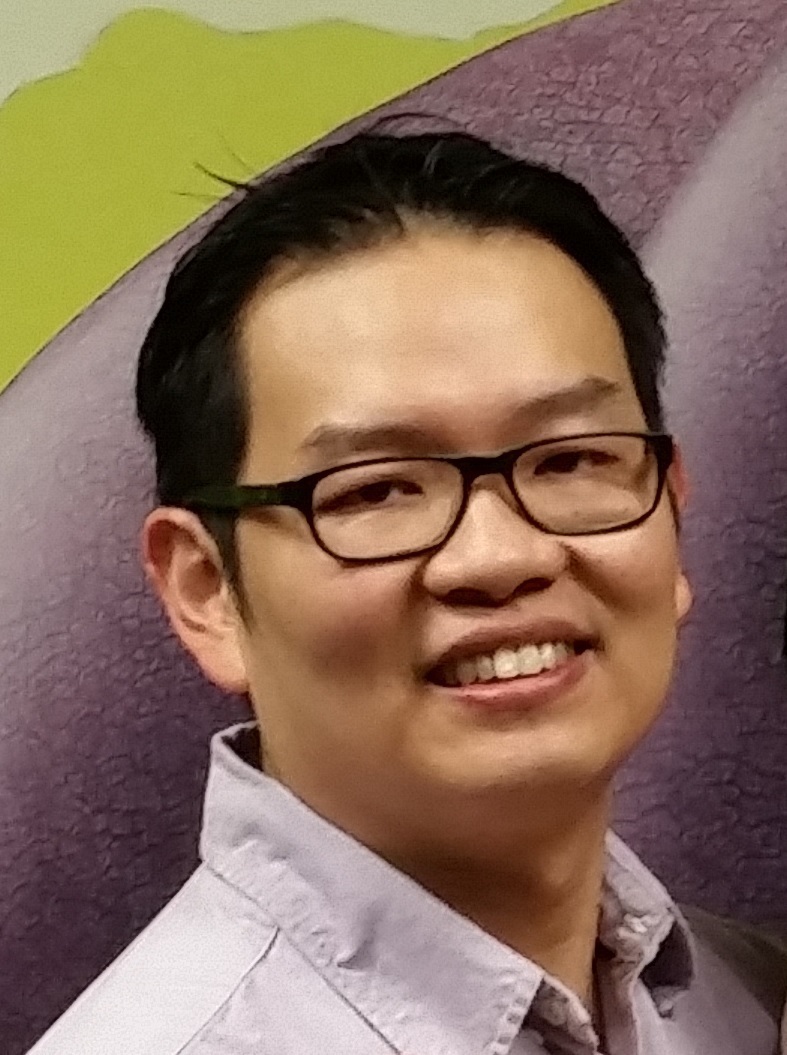}}]{Boon Loong Ng} received the Bachelor of Engineering (Electrical and Electronic) degree and the Ph.D. degree in Engineering from the University of Melbourne, Australia, in 2001 and in 2007, respectively. He currently holds the position of Senior Research Director with Samsung Research America – Standards \& Mobility Innovation (SMI) Lab in Plano, Texas. He contributed to 3GPP RAN L1/L2 standardizations of LTE, LTE-A, LTE-A Pro, and 5G NR technologies from the period of 2008 to 2018. He holds over 60 USPTO-granted patents on LTE/LTE-A/LTE-A Pro/5G and more than 100 patent applications globally. Since 2018, he has been leading an R\&D team that develops system and algorithm design solutions for commercial 5G and Wi-Fi technologies.
\end{IEEEbiography}

\begin{IEEEbiography}[{\includegraphics[width=1in,height=1.25in,clip,keepaspectratio]{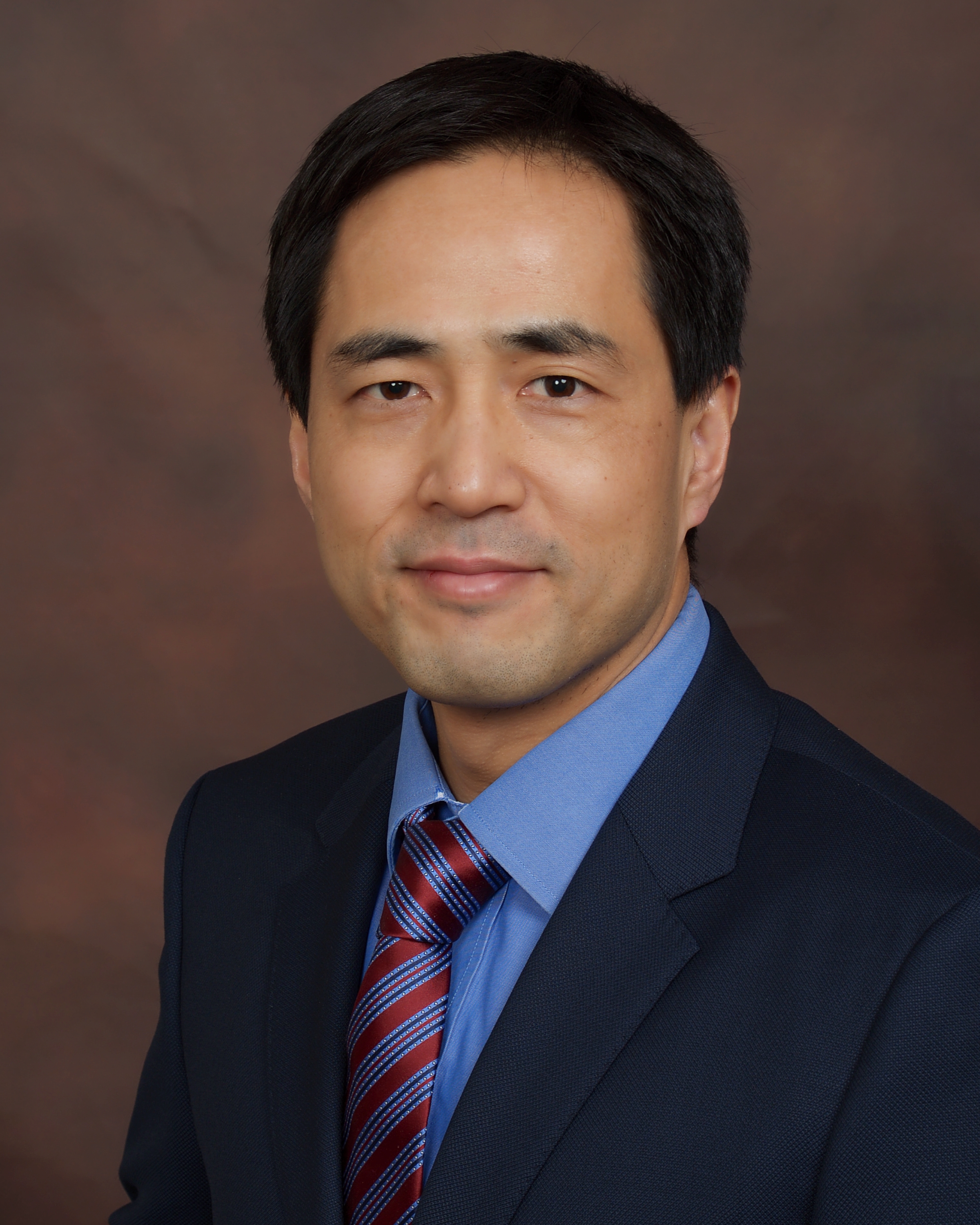}}]{Jianzhong (Charlie) Zhang} (S'00-M'03-SM'09-F'16) received the Ph.D. degree from the University of Wisconsin, Madison WI, USA. He is currently a Senior Vice President at Samsung Research America, where he leads research, prototyping, and standardization for 5G/6G and other wireless systems. He is also a Corporate VP and head of the  global 6G team at Samsung Research. He is currently serving as the ATIS North America Next-G Alliance Full Member Group Vice Chair. Previously, he was the Board Chair of the FiRa Consortium  from May 2019 to May 2023, and the Vice Chairman of the 3GPP RAN1 working group from 2009 to 2013, where he led development of LTE and LTE-Advanced technologies. He worked for Nokia Research Center and Motorola Mobility for 6 years before joining Samsung in 2007.  He received his Ph.D. degree from the University of Wisconsin, Madison.  Dr. Zhang is a Fellow of IEEE.
\end{IEEEbiography}

\end{document}